\documentclass[10pt]{article} 
\usepackage[accepted]{tmlr}

\usepackage{amsmath,amsfonts,bm}

\def\eqref#1{equation~\ref{#1}}

\def\1{\bm{1}}

\def\rvv{{\mathbf{v}}}

\def\rvx{{\mathbf{x}}}
\def\rvy{{\mathbf{y}}}
\def\rvz{{\mathbf{z}}}

\def\mA{{\bm{A}}}
\def\mB{{\bm{B}}}
\def\mC{{\bm{C}}}

\def\mF{{\bm{F}}}

\def\mH{{\bm{H}}}
\def\mI{{\bm{I}}}
\def\mJ{{\bm{J}}}

\def\mP{{\bm{P}}}

\def\mU{{\bm{U}}}
\def\mV{{\bm{V}}}
\def\mW{{\bm{W}}}

\def\mSigma{{\bm{\Sigma}}}

\DeclareMathAlphabet{\mathsfit}{\encodingdefault}{\sfdefault}{m}{sl}
\SetMathAlphabet{\mathsfit}{bold}{\encodingdefault}{\sfdefault}{bx}{n}

\def\gD{{\mathcal{D}}}

\def\gN{{\mathcal{N}}}

\usepackage{hyperref}
\usepackage{url}
\usepackage[utf8]{inputenc} 
\usepackage[T1]{fontenc}    
\usepackage{amsfonts}       
\usepackage{nicefrac}       
\usepackage{microtype}      

\usepackage{graphicx}
\usepackage{amsmath}
\usepackage{amsthm}
\usepackage{thm-restate}
\usepackage{svg}
\usepackage[ruled, linesnumbered, noend]{algorithm2e}
\usepackage{multirow}
\usepackage{caption}
\usepackage{subcaption}
\usepackage[nameinlink]{cleveref}
\usepackage{xcolor}

\title{GSURE-Based Diffusion Model \\ Training with Corrupted Data}

\author{\name Bahjat Kawar\thanks{Equal contribution.} \email bahjat.kawar@cs.technion.ac.il \\
        \addr Department of Computer Science
      \AND
      \name Noam Elata$^*$ \email noamelata@campus.technion.ac.il \\
      \addr Department of Electrical and Computer Engineering
      \AND
      \name Tomer Michaeli \email tomer.m@ee.technion.ac.il \\
      \addr Department of Electrical and Computer Engineering
      \AND
      \name Michael Elad \email elad@cs.technion.ac.il \\ 
      \addr Department of Computer Science \\ \\
      \addr Technion - Israel Institute of Technology, Haifa, Israel}

\def\month{04}  
\def\year{2024} 
\def\openreview{\url{https://openreview.net/forum?id=BRl7fqMwaJ}} 

\theoremstyle{plain}

\theoremstyle{definition}

\theoremstyle{remark}

\newcommand{\rveps}[0]{\boldsymbol{\epsilon}}

\newcommand{\xbar}[0]{\bar{\mathbf{x}}}
\newcommand{\ybar}[0]{\bar{\mathbf{y}}}
\newcommand{\zbar}[0]{\bar{\mathbf{z}}}

\newcommand{\bbR}[0]{\mathbb{R}}

\newcommand{\bbE}[0]{\mathbb{E}}

\begin{document}

\maketitle

\begin{abstract}
Diffusion models have demonstrated impressive results in both data generation and downstream tasks such as inverse problems, text-based editing, classification, and more.
However, training such models usually requires large amounts of clean signals which are often difficult or impossible to obtain.
In this work, we propose a novel training technique for generative diffusion models based only on corrupted data.
We introduce a loss function based on the Generalized Stein's Unbiased Risk Estimator (GSURE), and prove that under some conditions, it is equivalent to the training objective used in fully supervised diffusion models.
We demonstrate our technique on face images as well as Magnetic Resonance Imaging (MRI), where the use of undersampled data significantly alleviates data collection costs.
Our approach achieves generative performance comparable to its fully supervised counterpart without training on any clean signals.
In addition, we deploy the resulting diffusion model in various downstream tasks beyond the degradation present in the training set, showcasing promising results\footnote{Our code is available at \url{https://github.com/bahjat-kawar/gsure-diffusion}.}. 
\end{abstract}

\section{Introduction}
Denoising diffusion probabilistic models (DDPMs)~\citep{sohl2015deep, ho2020denoising, song2019generative}, or diffusion models for short, are a family of generative models that has recently risen to prominence.
They have achieved state-of-the-art performance in image generation~\citep{song2020score, vahdat2021score, dhariwal2021diffusion, latent_diffusion, kim2022refining}, as well as impressive generative modeling capabilities in other modalities~\citep{ho2022imagen, singer2023makeavideo, kong2021diffwave, popov2021grad, gong2022diffuseq, li2022diffusionlm, tevet2022human}, including protein structures~\citep{watson2022broadly, qiao2022dynamic, schneuing2022structure, yim2023se} and medical data~\citep{song2023solving, chung2022score, jalal2021robust, xie2022measurement, chung2023solving, li2023descod, adib2023synthetic}.
The prowess and flexibility of DDPMs have enabled their profound impact on downstream applications~\citep{kawar2022denoising, theis2022lossy, blau2022threat, pinaya2022fast, wyatt2022anoddpm, kawar2023imagic, zimmermann2021score}.

Training a diffusion model to learn an unknown data distribution is a complex task.
It usually requires training parameter-heavy neural networks on large amounts of pristine data.
For instance, diffusion models' success in image generation was in part enabled by large curated datasets, containing millions or even billions of images~\citep{imagenet, schuhmann2022laion}.
However, such large-scale datasets of pristine samples may often be expensive, difficult, or even impossible to obtain, especially in the medical domain~\citep{Mullainathan2022}. Interest in the degraded data setting has risen in recent years~\cite{xiangddm, aali2023solving, daras2024ambient}, yet existing solutions address only specific cases.
In this work, we present GSURE-Diffusion, a method for training generative diffusion models based on data corrupted by linear degradations and Gaussian noise.
This setting can make data collection for deep learning significantly faster and less expensive.

\begin{figure}
    \centering
    \includegraphics[width=\textwidth]{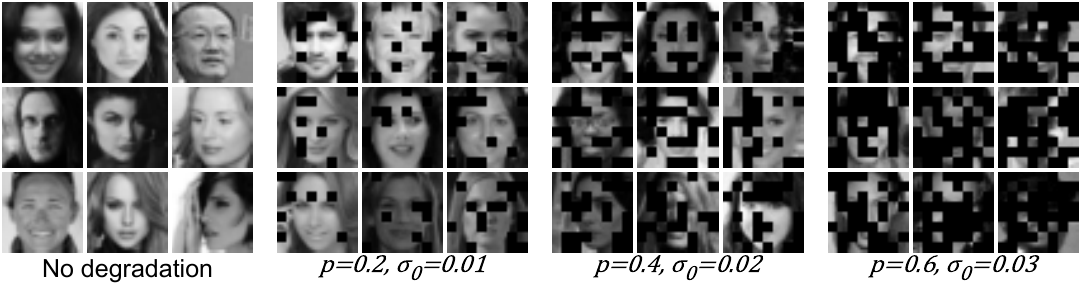}
    \caption{Training sets samples of the different degradation settings in CelebA~\citep{liu2015celeba} experiments.}
    \label{fig:celeba_training}
\end{figure}

GSURE-Diffusion operates on a datasets of noisy linear measurements of signals, and assumes the signal acquisition process is randomized within a fixed general structure, which is the case in many real-world applications.
In this setting, we present a novel loss function to learn the underlying data distribution.
First, we use the Singular Value Decomposition (SVD) of the degradation operators to decouple the measurement equation, following DDRM~\citep{kawar2022denoising}. This transformation simplifies the degradation into a masking operation.
Then, we add synthetic noise to the SVD-transformed measurements, likening them to the noisy samples used in the DDPM~\citep{ho2020denoising} framework.
Finally, we use the ensemble version of the Generalized Stein's Unbiased Risk Estimator (GSURE)~\citep{aggarwal2022ensure, eldar2008generalized} to learn to denoise samples without access to ground-truth clean signals. Simply put, GSURE's mathematical formulation allows us to makes use of the undamaged data within the corrupted image. 
Our proposed GSURE-based loss function is general to all randomized linear measurement settings, and we prove its equivalence to the fully supervised denoising diffusion loss under some conditions.

To empirically evaluate our technique, we apply it on a downsized grayscale version of CelebA~\citep{liu2015celeba}, a dataset of celebrity face images.
We train a GSURE-Diffusion model on noisy images with randomly missing patches, and compare its generative output with an oracle model that trained on the full clean images.
We observe that GSURE-Diffusion results in comparable generative performance to the oracle, despite having trained solely on corrupted data.

Furthermore, we use GSURE-Diffusion for Magnetic Resonance Imaging (MRI), a ubiquitous medical imaging modality providing vital diagnostic information.
Generative models for MRI are usually trained on datasets of fully sampled MR images, which can be expensive to obtain.
In contrast, we train a generative model on noisy undersampled data obtained from accelerated MRI scans.
By utilizing GSURE-Diffusion, we train a model with performance comparable to an oracle version while significantly reducing the time and resources required to collect the training dataset.
Then, we showcase the flexibility of the generative diffusion model we obtain. We use this model for accelerated MRI reconstruction, extending its applicability to acceleration factors beyond those encountered during training. Moreover, this model may serve as a foundational framework for addressing various tasks. We demonstrate its capabilities in MRI reconstruction for different subsampling strategies and uncertainty quantification, and compare its results with an oracle version.

In conclusion, we present GSURE-Diffusion, a novel method for training generative diffusion models based on corrupted measurements of the underlying data.
We demonstrate the capabilities of the method through several experiments, and show its applicability to real-world problems.
We hope that GSURE-Diffusion will facilitate future work on generative modelling for challenging settings, generalizing for more complex scenarios and various modalities.


\section{Background}
\subsection{Denoising Diffusion Probabilistic Models}
Denoising Diffusion Probabilistic Models (DDPMs)~\citep{ho2020denoising} are a family of generative models that learn a distribution $p_\theta(\rvx)$, approximating a data distribution $q(\rvx)$ from a dataset $\gD$ of samples.
DDPMs follow a Markov chain structure ${\rvx_T \rightarrow \rvx_{T-1} \rightarrow \dots \rightarrow \rvx_1 \rightarrow \rvx_0}$ that reverses a forward noising process from $\rvx_0$ to $\rvx_T$.
In the forward process, $\rvx_0$ is set to be $\rvx$, and
the intermediate variables $\rvx_t$ are defined by $q^{(t)}(\rvx_t | \rvx_{t-1})$, usually chosen to be a simple Gaussian $\gN\left(\sqrt{1 - \beta_t} \rvx_{t-1}, \beta_t \mI\right)$.
This leads to a useful property, ${q^*(\rvx_t | \rvx_0) = \gN\left(\sqrt{\bar{\alpha}_t} \rvx_0, \left(1 - \bar{\alpha}_t\right) \mI \right)}$ with $\bar{\alpha}_t = \prod_{s=1}^{t} \left(1 - \beta_t\right)$, which facilitates model training.
In the reverse and more challenging process,
the learned distribution $p_\theta^{(t)}(\rvx_{t-1} | \rvx_t)$ is also modeled as Gaussian, with a learned mean dependent on a neural network $f_\theta^{(t)}(\rvx_t)$ and a fixed~\citep{ho2020denoising} or learned~\citep{nichol2021improved} covariance.

The diffusion model $f_\theta^{(t)}(\rvx_t)$ is trained to optimize an evidence lower bound (ELBO) on the log-likelihood objective~\citep{sohl2015deep}.
The ELBO can be simplified into the following denoising objective:
\begin{equation}
\label{eq:diffusion-loss}
    \sum_{t=1}^{T} \gamma_t \bbE 
    \left[ \left\lVert  f_\theta^{(t)}(\rvx_t) - \rvx_0 \right\rVert_2^2 \right],
\end{equation}
where the $\gamma_t$ assign weights to different $t$ (different noise levels), and the expectation is over some ${\rvx_t \sim q(\rvx_t | \rvx_0)}$, ${\rvx_0 \sim q(\rvx)}$. Please refer to~\citep{ho2020denoising, song2020denoising} for derivations.
After training, diffusion models synthesize data by starting with a sample ${\rvx_T \sim \gN(0, \mI)}$, following the learned distributions $p_\theta^{(t)}$ along the Markov chain, sampling from each, and outputting $\rvx_0$ as the final sample.
Diffusion models have had incredible success in image generation~\citep{song2020score, vahdat2021score, dhariwal2021diffusion, latent_diffusion, kim2022refining}, as well as generation of other modalities~\citep{ho2022imagen, singer2023makeavideo, kong2021diffwave, popov2021grad, gong2022diffuseq, li2022diffusionlm, tevet2022human}.
They have also been deployed in a myriad of related tasks~\citep{kawar2022denoising, theis2022lossy, blau2022threat, pinaya2022fast, wyatt2022anoddpm, kawar2023imagic, zimmermann2021score}.
In this work, we revisit DDPMs and seek a way to train them using only corrupted data.

\begin{figure}
    \centering
    \includegraphics[width=\textwidth]{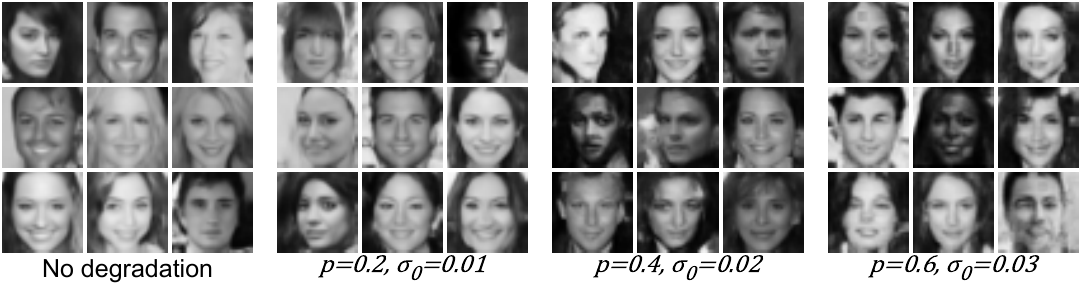}
    \caption{Generated samples (with $50$ DDIM~\citep{song2020denoising} steps) from models trained on different degradation settings in CelebA~\citep{liu2015celeba} experiments.}
    \label{fig:celeba_gen}
\end{figure}

\subsection{Generalized Stein's Unbiased Risk Estimator (GSURE)}
Given noisy measurements $\rvy = \rvx + \rvz$ (where $\rvx, \rvy, \rvz \in \bbR^{n}$) with noise $\rvz \sim \gN(0, \sigma^2 \mI)$, and a function $f(\rvy)$ aiming to estimate $\rvx$ from $\rvy$,
Stein's unbiased risk estimator (SURE)~\citep{stein1981estimation} is an unbiased estimator for the mean squared error (MSE) of $f(\rvy)$, formulated as
\begin{equation}
\label{eq:sure}
    \bbE\left[ \left\lVert f(\rvy) - \rvx \right\rVert_2^2 \right] = 
    \bbE\left[ \left\lVert f(\rvy) - \rvy \right\rVert_2^2 \right]
    + 2 \sigma^2 \bbE\left[ \nabla_{\rvy} \cdot f(\rvy) \right]
    - n \sigma^2.
\end{equation}
Crucially, SURE provides the ability to estimate the MSE of a denoiser $f(\rvy)$ without access to clean signals $\rvx$.
As a result, many researchers have used SURE for unsupervised learning of denoisers~\citep{zhang1998adaptive, blu2007sure, metzler2018unsupervised, soltanayev2018training, nguyen2020hyperspectral, jo2021rethinking}, even extending to training diffusion models based on (fully sampled) noisy data~\citep{xiangddm, aali2023solving}.

In the context of inverse problems, SURE has been generalized for corrupted measurements beyond additive white Gaussian noise~\citep{eldar2008generalized}.
The Generalized SURE (GSURE) considers $\rvy = \mH \rvx + \rvz$ (where $\rvx \in \bbR^{n}$, $\mH \in \bbR^{m \times n}$, $\rvy, \rvz \in \bbR^{m}$, and $\rvz \sim \gN(0, \mC)$), and a function $f(\rvy)$ estimating $\rvx$.
In this case, GSURE provides an unbiased estimate for the projected MSE:
\begin{equation}
\label{eq:gsure}
    \bbE\left[ \left\lVert \mP \left(f(\rvy) - \rvx \right) \right\rVert_2^2 \right] =
    \bbE\left[ \left\lVert \mP \left(f(\rvy) - \rvx_\mathrm{ML} \right) \right\rVert_2^2 \right]
    + 2 \bbE\left[ \nabla_{\mH^\top \mC^{-1} \rvy} \cdot \mP f(\rvy) \right]
    + c,
\end{equation}
where $\mH^\dagger$ is the Moore-Penrose pseudo-inverse of $\mH$, $\mP = \mH^\dagger \mH$ is a projection matrix onto the range-space of $\mH$, $\rvx_\mathrm{ML} = \left(\mH^\top \mC^{-1} \mH\right)^\dagger \mH^\top \mC^{-1} \rvy$, and $c$ is a constant that does not depend on $f(\rvy)$.
Several works have utilized GSURE for solving inverse problems by training only on corrupted measurements~\citep{metzler2018unsupervised, zhussip2019training, liu2020rare, abu2022image}.
However, when $\mH$ causes significant information loss, the projected MSE stops being a good proxy for the full MSE.
Ensemble SURE (ENSURE)~\citep{aggarwal2022ensure} learns from a dataset of measurements, each corrupted by a different operator $\mH$.
Therefore, the expectation over the projected MSE is taken over $\mH$ as well as the data and noise. This constitutes a more accurate proxy for the full MSE without relying on clean signals.
In this work, we extend the ENSURE framework for training a diffusion model using corrupted data.

\section{GSURE-Diffusion: Mathematical Formulation}
\subsection{Problem Formulation}
\label{sec:problem-formulation}
We are interested in training a generative diffusion model that can sample from an unknown data distribution $q(\rvx)$. However, we only have access to a dataset $\gD$ of corrupted measurements
\begin{equation}
\label{eq:inverse-problem}
    \rvy = \mH \rvx + \rvz,
\end{equation}
where $\rvy \in \bbR^m$, $\rvx \in \bbR^n$, $\mH \in \bbR^{m \times n}$, and $\rvz \sim \gN(0, \sigma_0^2 \mI)$ is additive white Gaussian noise (AWGN).\footnote{Our method can also handle anisotropic uncorrelated noise. We only consider AWGN to simplify notations.}
\autoref{eq:inverse-problem} refers to a single instance of an ideal image and its corresponding measurement, and more generally, different measurements $\rvy$ in the dataset may relate to different signals $\rvx$, different degradation procedures $\mH$, and different noise realizations $\rvz$. We assume $\rvx$, $\rvz$, and $\mH$ are randomly and independently sampled from their respective distributions.

In order to decouple the mathematical relationship between the observed measurements and the underlying data, we follow~\citep{kawar2022denoising} and utilize the singular value decomposition (SVD) of $\mH$,
\begin{equation}
\label{eq:svd}
    \mH = \mU \mSigma \mV^\top,
\end{equation}
where $\mU \in \bbR^{m \times m}$ and $\mV \in \bbR^{n \times n}$ are orthogonal matrices, and $\mSigma \in \bbR^{m \times n}$ is a rectangular diagonal matrix containing the singular values of $\mH$.
We define $\xbar = \mV^\top \rvx$, $\ybar = \mSigma^\dagger \mU^\top \rvy$, and $\zbar = \mSigma^\dagger \mU^\top \rvz$. Using these definitions and the SVD, \autoref{eq:inverse-problem} becomes
\begin{equation}
\label{eq:transformed-problem}
    \ybar = \mP \xbar + \zbar,
\end{equation}
where $\mP = \mSigma^\dagger \mSigma$ is a diagonal subsampling matrix with zeroes and ones, and $\zbar \sim \gN(0, \sigma_0^2 \mSigma^\dagger \mSigma^{\dagger\top})$ constitutes anisotropic uncorrelated Gaussian noise.

We make the following assumptions on the training dataset $\gD$:
(i) The sampling matrices $\mH$ and noise levels $\sigma_0$ are known;
(ii) All matrices $\mH$ share the same right-singular vectors $\mV^\top$; and 
(iii) The different $\mH$ across the dataset jointly cover the signal space $\bbR^n$, \textit{i.e.}, $\bbE[\mP]$ is a positive definite matrix.\footnote{Under these notations, $\bbE[\mP]$ is measured for a fixed $\mV^\top$, and the values in $\mSigma$ are not necessarily ordered.}
These assumptions are satisfied in many real-world applications such as Magnetic Resonance Imaging (MRI) --
measurements are acquired in a similar fashion for all data points (upholding (ii)), and the subsampling pattern can be randomly chosen for each point (upholding (iii)). $\mH$ and $\sigma_0$ are derived from the physics of the signal acquisition procedure, thus upholding (i).

Under the transformed measurement equation presented in \autoref{eq:transformed-problem}, we aim to train a generative model for $\xbar$, which can easily translate to $\rvx$ through the orthogonal transformation $\rvx = \mV \xbar$.

\subsection{GSURE-Based Denoising Diffusion Loss Function}
\label{sec:gsure-loss}
In order to train a diffusion model for $\xbar$, we aim to obtain noisy training samples $\xbar_t$ that satisfy the marginal distribution $q^*(\xbar_t | \xbar) = \gN\left(\sqrt{\bar{\alpha}_t} \xbar, \left(1 - \bar{\alpha}_t\right) \mI \right)$, as in traditional diffusion models.
However, we only have access to corrupted measurements $\ybar$ as in \autoref{eq:transformed-problem}.
For a given $t$, we perturb these measurements with additional noise according to
\begin{equation}
\label{eq:noise-add}
    \xbar_t = \sqrt{\bar{\alpha}_t} \ybar + \left(\left(1 - \bar{\alpha}_t\right) \mI - \bar{\alpha}_t \sigma_0^2 \mSigma^\dagger \mSigma^{\dagger\top}\right)^{\frac{1}{2}} \rveps_t,
\end{equation}
where $\rveps_t \sim \gN\left(0, \mI\right)$ is independently sampled.
Intuitively, $\sqrt{\bar{\alpha}_t} \ybar$ includes noise with a diagonal covariance $\bar{\alpha}_t \sigma_0^2 \mSigma^\dagger \mSigma^{\dagger\top}$.
We increase the noise level in each entry by an appropriate amount to reach a variance of $1 - \bar{\alpha}_t$ in all entries. Increasing the noise to the desired level is possible as long as $\left(\left(1 - \bar{\alpha}_t\right) \mI - \bar{\alpha}_t \sigma_0^2 \mSigma^\dagger \mSigma^{\dagger\top}\right)$ is a positive-semi-definite (PSD) matrix. Because $\bar{\alpha}_t$ is monotonically decreasing w.r.t $t$, by setting the beginning of the noise schedule such that $\left(\left(1 - \bar{\alpha}_t\right) \mI - \bar{\alpha}_t \sigma_0^2 \mSigma^\dagger \mSigma^{\dagger\top}\right)$ is PSD, we ensure that this covariance matrix is PSD for all timesteps.
This way, we obtain samples $\xbar_t$ suitable for training a diffusion model, as they follow the marginal distribution
\begin{equation}
\label{eq:xbart-marginal}
    q(\xbar_t | \xbar, \mP) = \gN\left(\sqrt{\bar{\alpha}_t} \mP \xbar, \left(1 - \bar{\alpha}_t\right) \mI \right).
\end{equation}
This resembles the ideal distribution of training samples $q^*(\xbar_t | \xbar)$, differing only in the mean value for entries dropped by $\mP$.
In the following, we derive a loss function that uses $\xbar_t$ satisfying \autoref{eq:xbart-marginal},
and utilizes an expectation over $\xbar$, $\zbar$, and $\mP$ yielding samples $\xbar_t \sim \bbE_{\mP}[q(\xbar_t | \xbar, \mP)]$.
This results in an estimate for denoising ideal samples from $q^*(\xbar_t | \xbar)$.
The estimate assumes that the model has the ability to generalize for samples from $q^*(\xbar_t | \xbar) = q(\xbar_t | \xbar, \mP= \mI)$, despite having trained only on signals with an undersampled $\mP$. We validate our model's ability to denoise samples from $q^*(\xbar_t | \xbar)$ in \autoref{sec:generalization}.

Ideally, we would like to train a diffusion model $f_\theta^{(t)}(\xbar_t)$ using the traditional denoising diffusion loss function in \autoref{eq:diffusion-loss}.
However, as we only have access to undersampled measurements, we consider a weighted expected projected MSE objective (similar to ENSURE~\citep{aggarwal2022ensure}):
\begin{equation}
\label{eq:projected-diffusion-loss}
    \sum_{t=1}^{T} \gamma_t \bbE \left[ \left\lVert \mW \mP \left(  f_\theta^{(t)}(\xbar_t) - \xbar \right) \right\rVert_2^2 \right],
\end{equation}
where the expectation is taken over ${\xbar_t \sim q(\xbar_t | \xbar, \mP)}$, $\xbar \sim q(\xbar)$, $\mP$ is independently sampled and $\mW = \bbE[\mP]^{-\frac{1}{2}} \succ 0$ (positive definite). The weight matrix $\mathbf{W}$ is placed in \autoref{eq:projected-diffusion-loss} for balancing the effect of the projections $\mathbf{P}$, in the case where elements of $\mathbf{P}$ do not have equal probability. In practice, this expectation is realized through $\ybar$ sampled from the dataset $\gD$ and \autoref{eq:noise-add}.
\begin{restatable}{proposition}{ensemble}
\label{thm:ensemble}
For $\rvx \sim q(\rvx)$, $\xbar = \mV^\top \rvx$, $\xbar_t$ sampled from \autoref{eq:xbart-marginal}, and the diagonal weight matrix $\mW = \bbE[\mP]^{-\frac{1}{2}} \succ 0$ (positive definite), if  $\mP$ and $\left(f_\theta^{(t)}(\xbar_t) - \xbar\right)$ are statistically independent, then \autoref{eq:projected-diffusion-loss} equals \autoref{eq:diffusion-loss}. 
\end{restatable}
We place the proof in \autoref{sec:proofs}.
The proof relies on the aforementioned assumptions made in ENSURE~\citep{aggarwal2022ensure}. 
We assess the validity of this assumption for our trained networks in \autoref{sec:independence}, and find that the assumption approximately holds for half the range of timesteps used in the denoiser. Nevertheless, even though the assumption is violated for some timesteps, the algorithm's performance remains similar to the oracle model, as shown in \autoref{fig:mri_denoising}. 
The expected projected MSE term in \autoref{eq:projected-diffusion-loss} measures the squared error of the denoiser $f_\theta^{(t)}(\xbar_t)$ only in entries kept by $\mP$.
This fact makes the loss easier to measure, as we do not have access to the entries dropped by $\mP$.
However, we still cannot accurately measure this loss, because we lack access to noiseless signals $\mP \xbar$.
To mitigate this, we utilize GSURE to estimate \autoref{eq:projected-diffusion-loss} using only $\xbar_t$ with the loss
\begin{equation}
\label{eq:gsure-diff-loss}
    \sum_{t=1}^{T} \gamma_t
    \bbE \left[ \left\lVert \mW \mP \left(  f_\theta^{(t)}(\xbar_t) - \frac{1}{\sqrt{\bar{\alpha_t}}} \xbar_t \right) \right\rVert_2^2
    + 2 \lambda_t \left( \nabla_{\xbar_t} \cdot \mP \mW^2 f_\theta^{(t)}(\xbar_t) \right)
    + c \right],
\end{equation}
where $c$ is a constant that does not depend on $\theta$, and $\lambda_t$ is a scalar hyperparameter.
The expectation is over the same random variables from \autoref{eq:projected-diffusion-loss}.
\begin{restatable}{proposition}{unbiased}
\label{thm:unbiased}
For $\rvx \sim q(\rvx)$, $\xbar = \mV^\top \rvx$, $\xbar_t$ sampled from \autoref{eq:xbart-marginal}, and $\lambda_t = 1 - \bar{\alpha}_t$, it holds that \autoref{eq:gsure-diff-loss} equals \autoref{eq:projected-diffusion-loss}.
\end{restatable}
\newcommand*{\propositionCrefname}{Proposition}
We place the proof in the \autoref{sec:proofs}.
\autoref{thm:ensemble} and \autoref{thm:unbiased} present a principled method to train a denoising diffusion model based only on corrupted data $\ybar$.
By minimizing the loss function in \autoref{eq:gsure-diff-loss}, we obtain a trained diffusion model that can be utilized in the same fashion as a fully supervised one, for both generation and downstream applications.

When our proposed training scheme is applied in practice, we apply the following modifications for practical considerations: First, the expectation term in \autoref{eq:gsure-diff-loss} is replaced by an average over training batches. Second, $\frac{1}{\sqrt{\bar{\alpha_t}}} \xbar_t$ is replaced by $\ybar$ to alleviate the high variance of the loss function at little to no cost in terms of bias. Lastly, the divergence term is calculated using the Hutchinson’s trace estimator, which is an unbiased Monte Carlo estimator~\citep{ramani2008monte, hutchinson1989stochastic}. While this estimator introduces some noise into our equation, we find that the estimation error is negligible when using a numerically stable derivative calculation. We expand upon these pragmatic implementation details in \autoref{sec:pragmatic}.

\begin{table}[b]
  \caption{FID~\citep{fid} results for diffusion models trained on increasing levels of degradation for $32 \times 32$-pixel CelebA~\citep{liu2015celeba} images, with different DDIM~\citep{song2020denoising} steps at generation time. Models were trained with (top) or without (bottom) GSURE-Diffusion, on degraded data.}
  \label{tab:celeba-fid}
  \begin{center}
  \begin{tabular}{l l c c c c}
    \textbf{Training Scheme} &
    \textbf{Data Degradation} &
    $\mathbf{10}$ \textbf{Steps} &
    $\mathbf{20}$ \textbf{Steps} &
    $\mathbf{50}$ \textbf{Steps} &
    $\mathbf{100}$ \textbf{Steps} \\
    \hline
    \multirow{4}{*}{Regular}
    & No degradation (oracle) & $21.99$ & $13.09$ & $08.15$ & $06.84$ \\
    & $p=0.2$, $\sigma_0=0.01$ & $86.25$ & $108.56$ & $119.30$ & $123.08$ \\
    & $p=0.4$, $\sigma_0=0.02$ & $216.74$ & $229.28$ & $236.61$ & $240.03$ \\
    & $p=0.6$, $\sigma_0=0.03$ & $273.62$ & $279.17$ & $280.89$ & $281.60$ \\
    \hline
    \multirow{3}{*}{GSURE-Diffusion}
    & $p=0.2$, $\sigma_0=0.01$ & $18.77$ & $12.25$ & $08.84$ & $08.82$ \\
    & $p=0.4$, $\sigma_0=0.02$ & $19.26$ & $14.98$ & $14.03$ & $15.14$ \\
    & $p=0.6$, $\sigma_0=0.03$ & $34.51$ & $27.74$ & $26.42$ & $28.31$ \\
  \end{tabular}
  \end{center}
\end{table}


\section{Experiments}

In the following, we demonstrate the capabilities of our method for training a diffusion model using corrupted data. 
To obtain corrupted data, we simulate several corruptions on datasets containing clean images.
Then, we train a diffusion model based only on the corrupted data, and compare its results against an \textit{oracle} model (with identical training hyperparameters) which is trained on the pristine data with the traditional diffusion loss function from \autoref{eq:diffusion-loss}.

\subsection{Human Face Images}
\label{sec:celeba-exp}
To empirically evaluate GSURE-Diffusion, we apply it on $32\times32$-pixel grayscale face images from CelebA~\citep{liu2015celeba}. 
We simulate a corrupted measurement process by splitting the images into $4\times4$-pixel non-overlapping patches, and randomly erasing each patch with probability $p$. We also perturb the data with AWGN with standard deviation $\sigma_0$.
This degradation matches our assumptions in \autoref{sec:problem-formulation} (see \autoref{sec:data}), and we provide samples of it in \autoref{fig:celeba_training}.
We adapt the U-Net~\citep{unet} architecture from DDPM~\citep{ho2020denoising} to match the image dimensions,
and train diffusion models for increasing levels of degradation on the CelebA training set. Training hyperparameters and more details are provided in \autoref{sec:impl}.

After training, we generate images from the models using the deterministic DDIM~\citep{song2020denoising} sampling schedule. 
We measure the generative performance using the Fr\'{e}chet Inception Distance (FID)~\citep{fid} between $10000$ generated images and the CelebA validation set.
As can be seen in \autoref{tab:celeba-fid} and \autoref{fig:celeba_gen}, our GSURE-Diffusion models achieve generative performance comparable to the oracle model, despite having trained only on corrupted data.
As expected, the performance deteriorates when more severe degradations are applied to the training data.
We note that in the few timesteps regime, GSURE-Diffusion slightly outperforms the oracle model, possibly due to the regularization effect of training on corrupted measurements.
However, once we increase the number of timesteps used for generation,
the oracle model benefits from the more fine-grained process, whereas
GSURE-Diffusion struggles to do so.


\begin{figure}[t]
    \centering
    \includegraphics[width=0.92\textwidth]{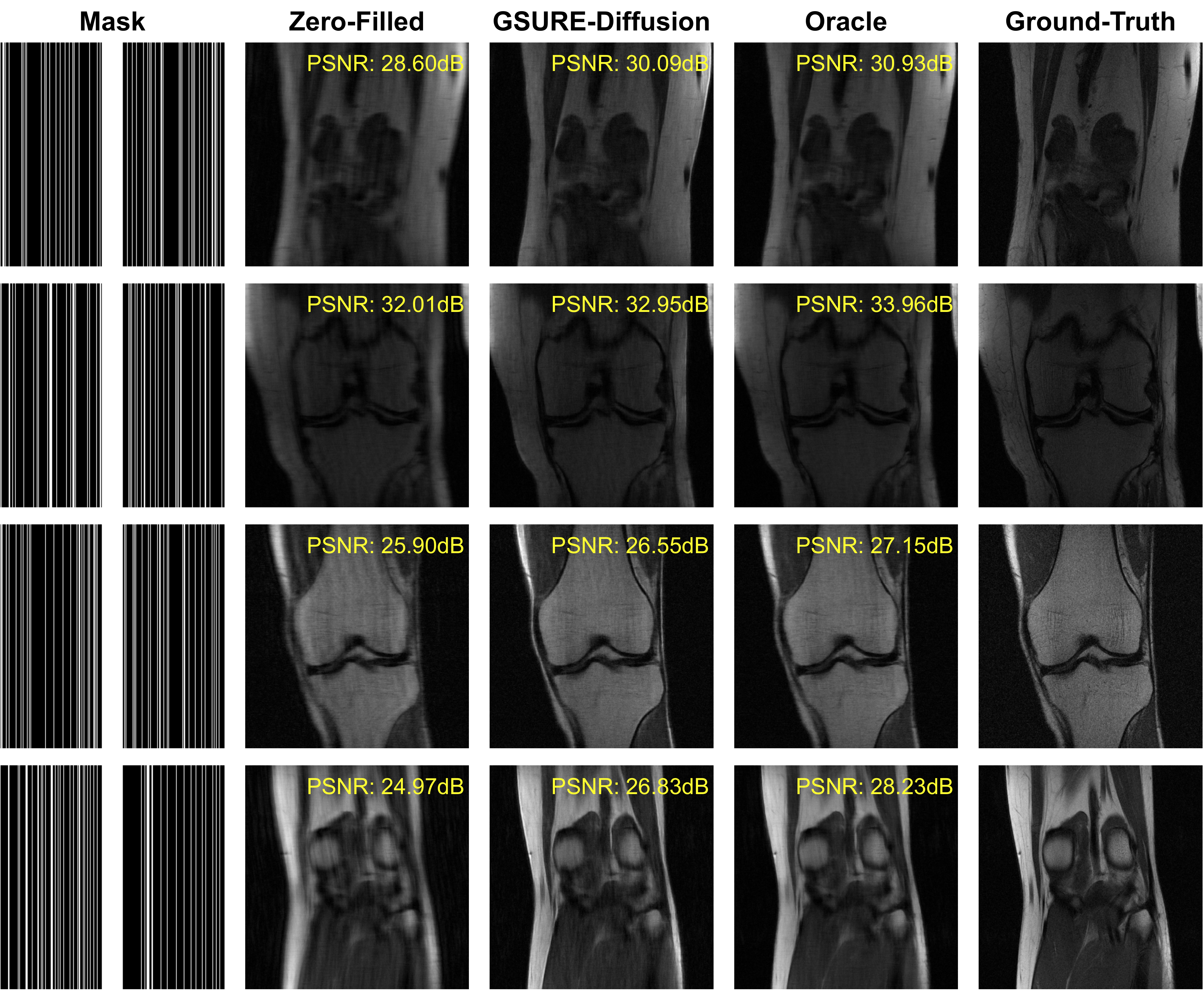}
    \caption{Accelerated MRI reconstruction results for $R=4$ and $\sigma_0=0.01$.}
    \label{fig:mri_recon_4}
\end{figure}

\subsection{Magnetic Resonance Imaging (MRI)}
\begin{figure}[t]
\label{fig:MRI-denoising}
    \centering
    \includegraphics[width=0.99\textwidth]{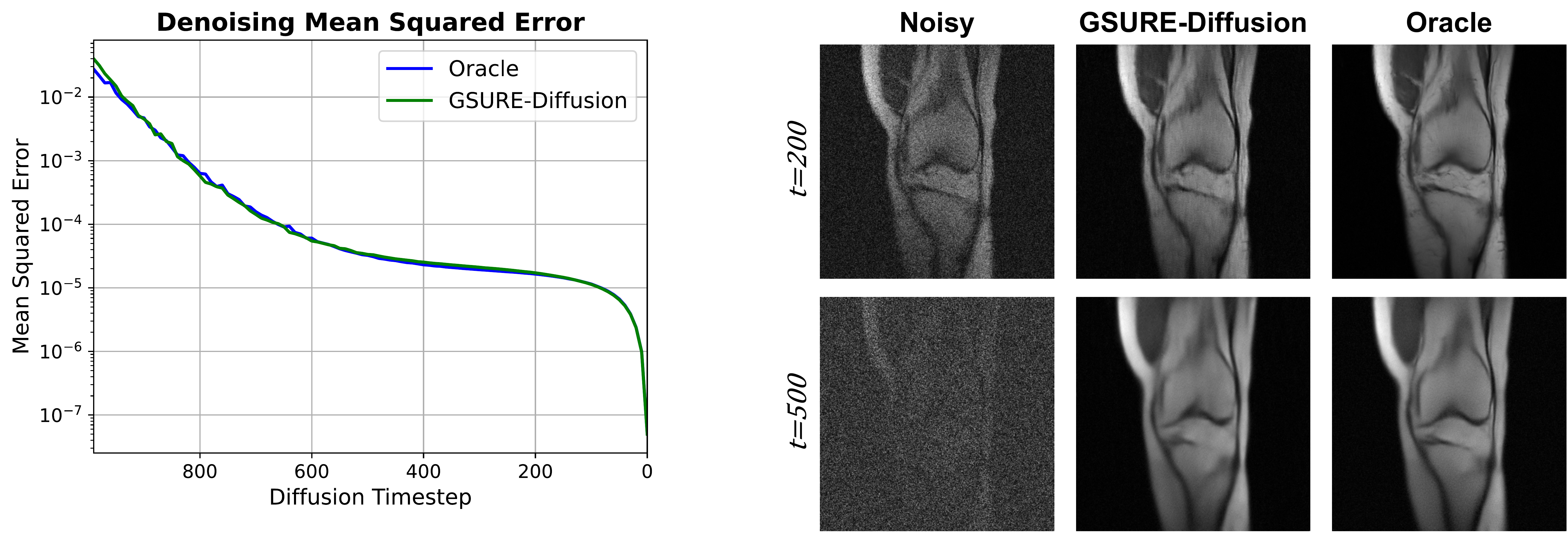}
    \caption{Left: Denoising MSE (on fully sampled noisy images) for the GSURE-Diffusion and oracle models across diffusion timesteps. Right: Qualitative denoising examples for both models.}
    \label{fig:mri_denoising}
\end{figure}

Magnetic Resonance Imaging (MRI) is a ubiquitous non-invasive medical imaging modality that can provide life-saving diagnostic information.
MRI measurements are obtained in the Fourier spectrum (also called k-space) of an object with magnetic fields.
However, measuring the entire k-space can be time-consuming and expensive.
Therefore, inferring the scanned object image based on partial, possibly randomized and noisy, k-space measurements from accelerated MRI scans is a highly relevant and challenging inverse problem, drawing considerable research attention~\citep{wang2016accelerating, hammernik2018learning, lee2018deep, han2019k, weiss2021pilot, wang2022one}.

The accelerated MRI procedure satisfies the assumptions we introduced in \autoref{sec:problem-formulation} (see \autoref{sec:data}), with the discrete Fourier transform as $\mV^\top$, and the random subsampling mask as $\mSigma$.
We use this fact and train a generative diffusion model for MRI based solely on accelerated scans.
We train on $24,853$ scanned slices from the fastMRI~\citep{knoll2020fastmri, zbontar2019fastmri} single-coil knee MRI dataset, center-cropped to a spatial size of $320\times320$.
The accelerated MRI subsampling process is simulated following~\citep{jalal2021robust} with an acceleration factor $R=4$, randomized subsampling of high frequencies, and AWGN with $\sigma_0=0.01$.
To facilitate training, real and imaginary elements of complex numbers are treated as separate input/output channels.
We use the same U-Net~\citep{unet, ho2020denoising} architecture from our CelebA~\citep{liu2015celeba} experiments, slightly modified to match the data dimensions,
and train a diffusion model on the corrupted measurements. A separate oracle model is trained with the same hyperparameters (detailed in \autoref{sec:impl}) on the fully sampled data.

To evaluate the validity of our approach, 
we measure the mean squared error (MSE) of both models in denoising $1024$ fully sampled MR images from the fastMRI~\citep{knoll2020fastmri, zbontar2019fastmri} validation set for different diffusion timesteps.
In \autoref{fig:mri_denoising}, we observe that the denoising ability of GSURE-Diffusion resembles that of the oracle model.
This supports our claim that GSURE-Diffusion can be a suitable replacement for the oracle model despite having trained exclusively on corrupted data.
To test this claim in real-world applications,
we substitute the oracle model for its GSURE-Diffusion counterpart in accelerated MRI reconstruction -- restoring MR images from noisy subsampled versions.
We use the general-purpose diffusion-based inverse problem solver DDRM~\citep{kawar2022denoising} for this task with $\eta=0$ and $100$ steps.
In \autoref{fig:mri_recon_4}, we observe that GSURE-Diffusion achieves similar performance to the oracle for $R=4$.
In all MRI experiments, we also present the ``zero-filled'' result (showing the MR image with zeroes in the missing frequencies) as a baseline.

Moreover, since we train a foundational MRI generative model, it is not restricted to the degradation setting present in its training data.
For instance, the model can generalize well for higher acceleration factors, as we show in \autoref{fig:mri_recon_gen}.
GSURE-Diffusion can also be utilized to reconstruct MR images corrupted by subsampling masks of different characteristics, such as 1-dimensional Gaussian random sampling and variable density Poisson disc sampling, as we show in \autoref{fig:mri_other}.
Furthermore, as a generative model, GSURE-Diffusion can provide uncertainty estimates for its outputs.
We follow~\citep{chung2022score} and quantify the uncertainty using the standard deviation of $8$ stochastic outputs made by the model.
We add synthetic Gaussian noise to MR images with $\sigma_0 = 0.4$,
and show uncertainty quantification results using GSURE-Diffusion and the oracle model in \autoref{fig:mri_uncertainty}.
This uncertainty quantification technique can potentially aid medical practitioners, providing clues towards anomalous regions in MRI scans. 
These results present evidence that a generative model trained on corrupted data can be deployed in various applications.
By loosening the requirements on the quality of the training data, we significantly reduce the cost of data acquisition for model training.



\begin{figure}[t]
    \centering
    \includegraphics[width=0.92\textwidth]{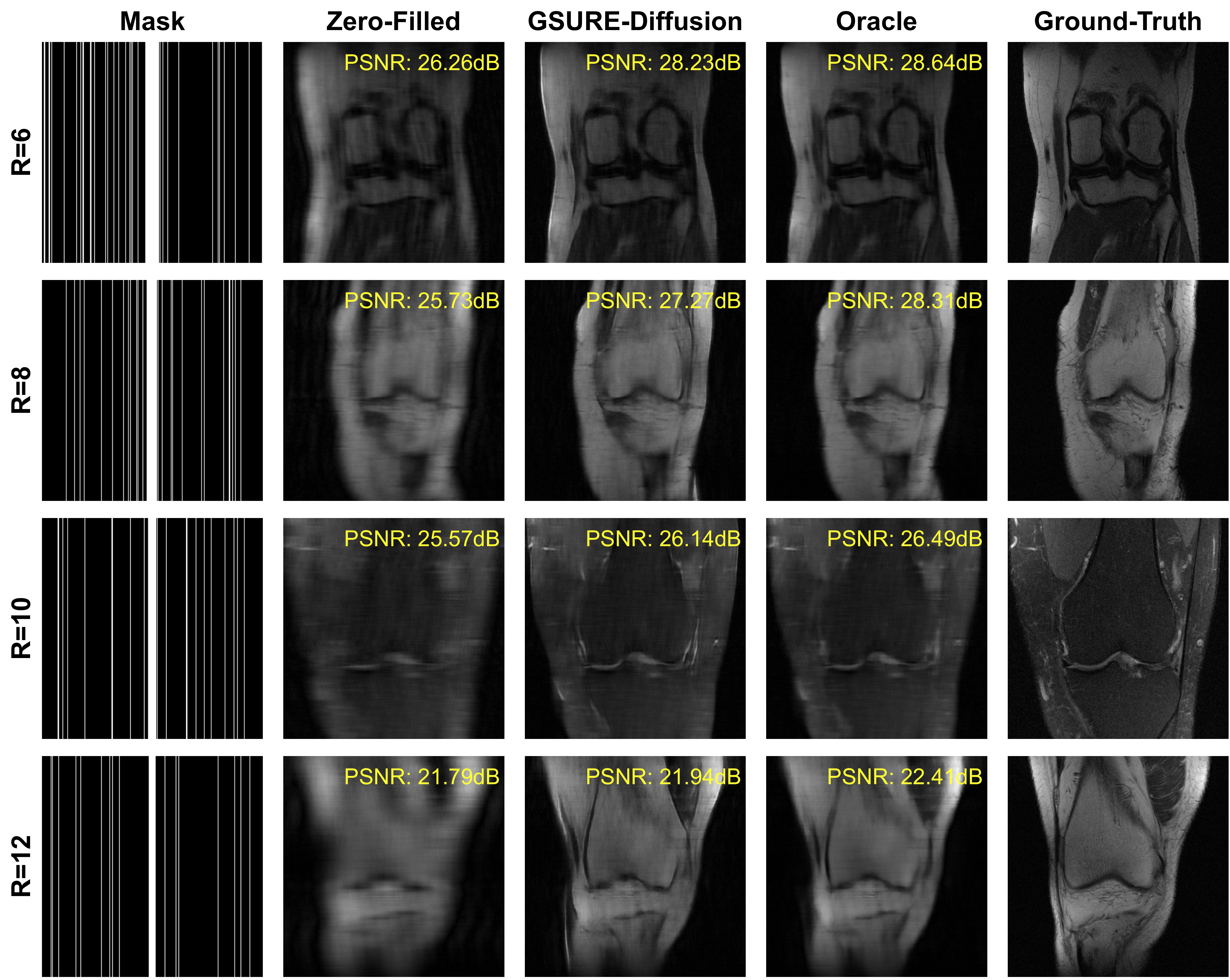}
    \caption{Accelerated MRI reconstruction results for $R \in \{6,8,10,12\}$ and $\sigma_0=0.01$. GSURE-Diffusion can generalize well across different acceleration factors.}
    \label{fig:mri_recon_gen}
\end{figure}

\section{Limitations}
While GSURE-Diffusion achieves impressive results, it suffers from a few limitations.
First, the assumptions over the available dataset and its acquisition procedure listed in \autoref{sec:problem-formulation} and \autoref{sec:gsure-loss}, which are satisfied in our example for accelerated MRI, may not always hold in other real-world scenarios.
Second, if the measurement noise $\zbar$ has high variance in at least one entry, the minimum noise level in the diffusion process, $\bar{\alpha}_1$, would also need to be high. This may lead to poor performance, as the final step of the generative diffusion process will need to clean significant noise.
Third, if the distribution of $\mP$ in the dataset is heavily biased towards certain regions, or has high variance in terms of $\mP$'s rank, this can break the diffusion model's ability to infer $\mP$ from $\xbar_t$, and thus degrade the generative performance.

The latter two limitations can be mitigated in future work by designing a diffusion model architecture that is explicitly aware of the ``mask'' $\mP$ for each input, and can handle $\xbar_t$ with a diagonal covariance with different values in the main diagonal.
Similar ideas (input masks and diagonal covariance) have been suggested in fully supervised diffusion modeling literature~\citep{bao2022estimating, gao2023masked}, and would be interesting to explore in the corrupted data setting.


\section{Related Work}

There has been a rich vein of works on unsupervised learning from datasets of corrupted data~\citep{lehtinen2018noise2noise, batson2019noise2self, hendriksen2020noise2inverse, chen2022robust}, including several SURE- and GSURE-based approaches~\citep{soltanayev2018training, nguyen2020hyperspectral, jo2021rethinking, metzler2018unsupervised, zhussip2019training, aggarwal2022ensure, abu2022image, liu2020rare}.
While they achieve impressive results, these efforts are mostly focused on learning a specific task. In contrast, our approach learns a foundational generative model, making it suitable for a wide range of applications. While generative modeling methods which learn only using corrupted data have been proposed in the past~\citep{mattei2019miwae, cheng2019misgan}, our method is the first to provide a holistic solution for learning from data corrupted beyond missing pixels, with general linear degradations.

\begin{figure}[t]
    \centering
    \includegraphics[width=0.85\textwidth]{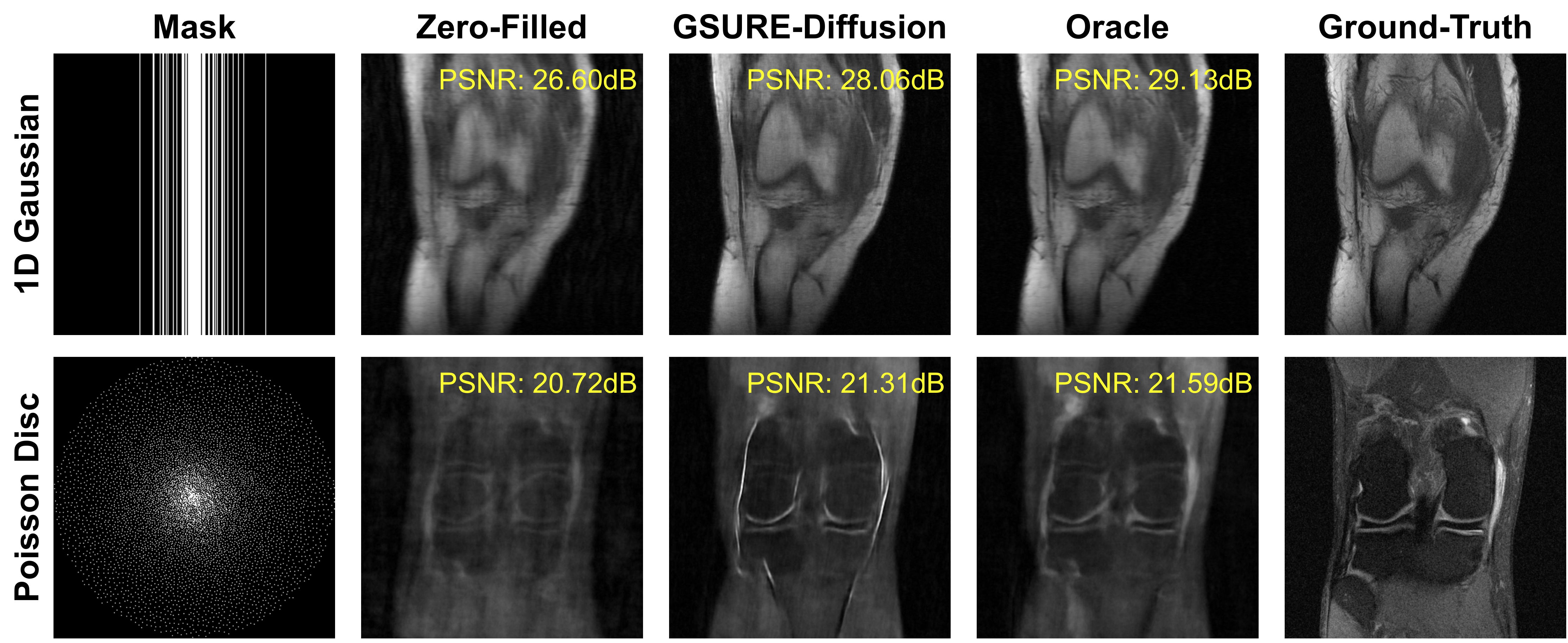}
    \caption{MRI reconstruction results for 1-dimensional Gaussian random sampling ($R=8$) and variable density Poisson disc sampling ($R=15$), both with $\sigma_0=0.01$. Our method generalizes well for different sampling schemes.}
    \label{fig:mri_other}
\end{figure}

\begin{figure}[t]
    \centering
    \hspace*{0.9em}\includegraphics[width=0.95\textwidth]{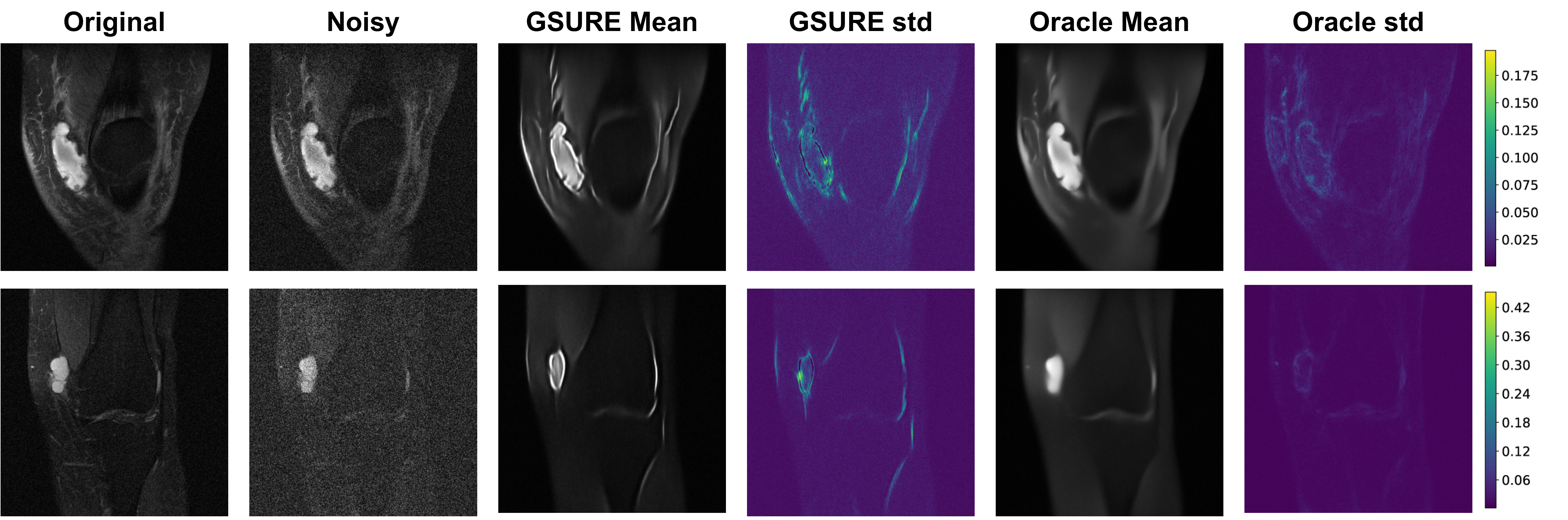}
    \caption{Uncertainty quantification for MR image denoising with GSURE-Diffusion and oracle models. Means and standard deviations are calculated for $8$ stochastic diffusion reconstructions.}
    \label{fig:mri_uncertainty}
\end{figure}

Recent advances in diffusion-based generative modeling~\citep{ho2020denoising, dhariwal2021diffusion} have enabled subsequent work to adapt these models for medical imaging (\textit{e.g.}, MRI)~\citep{song2023solving, chung2022score, jalal2021robust, xie2022measurement, chung2023solving}.
These diffusion models can serve a multitude of tasks, but can be expensive to train as they require fully sampled noiseless training data.
Notably, DDM$^2$~\citep{xiangddm} and the concurrent work of SURE-Score~\citep{aali2023solving} offer ways to train a diffusion model based on noisy data, which is often the case in practical settings. 
Similarly, a concurrent work by \cite{daras2024ambient} offers a method for training diffusion models based on data with missing pixels, showing improved results for downstream inpainting tasks.
However, collected data can also be undersampled or corrupted by other transformations.
Our proposed framework is more general, as it can handle both Gaussian noise and general linear corruptions.

Diffusion models have permeated various research areas, including online decision making~\citep{hsieh2023thompson}.
In that context, obtaining full noiseless data may often be impossible.
The authors of~\citep{hsieh2023thompson} suggest a diffusion loss function that learns a diffusion-based prior from noisy data with missing elements, similar to an image inpainting problem. They note that their proposed loss function could be of independent interest in future work.
Our loss function closely resembles theirs, albeit generalizing for linear corruptions beyond inpainting by utilizing the singular value decomposition (SVD).


\section{Conclusion}
We have introduced GSURE-Diffusion, a technique for training generative diffusion models based on data corrupted by linear degradations and additive Gaussian noise.
Using the SVD of the degradation operator, GSURE~\citep{eldar2008generalized}, and an ensemble of degradations~\citep{aggarwal2022ensure}, we introduce a novel training scheme that approximates the underlying data distribution from corrupted measurements.
We perform experiments on CelebA~\citep{liu2015celeba} to demonstrate the validity of our proposed framework.
Additionally, we show the applicability of GSURE-Diffusion to real-world problems by training a model on accelerated (undersampled and noisy) MRI scans from fastMRI~\citep{zbontar2019fastmri, knoll2020fastmri}.
We then use the resulting diffusion model to address several downstream tasks.

We hope that our proposed framework will enable future work to train generative models on similar problems such as multi-coil accelerated MRI, sparse-view Computed Tomography (CT), and more.
Future work may also address additional challenging measurement acquisition scenarios, as dictated by the data modality.
We emphasize the fact that our results are obtained on simulated corruptions. More extensive experiments on real data should be conducted before GSURE-Diffusion can be relied upon for diagnostics in clinical settings.

\section*{Acknowledgements}
This research was partially supported by the Council For Higher Education - Planning \& Budgeting Committee. Data used in the preparation of this article were obtained from \href{https://fastmri.med.nyu.edu/}{the NYU fastMRI Initiative database}~\cite{zbontar2019fastmri, knoll2020fastmri}. As such, NYU fastMRI investigators provided data but did not participate in analysis or writing of this report. A listing of NYU fastMRI investigators, subject to updates, can be found at \href{https://fastmri.med.nyu.edu/}{fastmri.med.nyu.edu}.
The primary goal of fastMRI is to test whether machine learning can aid in the reconstruction of medical images.

\newpage

\bibliography{refs}

\begin{thebibliography}{75}
\providecommand{\natexlab}[1]{#1}
\providecommand{\url}[1]{\texttt{#1}}
\expandafter\ifx\csname urlstyle\endcsname\relax
  \providecommand{\doi}[1]{doi: #1}\else
  \providecommand{\doi}{doi: \begingroup \urlstyle{rm}\Url}\fi

\bibitem[Aali et~al.(2023)Aali, Arvinte, Kumar, and Tamir]{aali2023solving}
Asad Aali, Marius Arvinte, Sidharth Kumar, and Jonathan~I Tamir.
\newblock Solving inverse problems with score-based generative priors learned from noisy data.
\newblock \emph{arXiv preprint arXiv:2305.01166}, 2023.

\bibitem[Abu-Hussein et~al.(2022)Abu-Hussein, Tirer, Chun, Eldar, and Giryes]{abu2022image}
Shady Abu-Hussein, Tom Tirer, Se~Young Chun, Yonina~C Eldar, and Raja Giryes.
\newblock Image restoration by deep projected {GSURE}.
\newblock In \emph{Proceedings of the IEEE/CVF Winter Conference on Applications of Computer Vision}, pp.\  3602--3611, 2022.

\bibitem[Adib et~al.(2023)Adib, Fernandez, Afghah, and Prevost]{adib2023synthetic}
Edmond Adib, Amanda Fernandez, Fatemeh Afghah, and John~Jeff Prevost.
\newblock Synthetic ecg signal generation using probabilistic diffusion models.
\newblock \emph{arXiv preprint arXiv:2303.02475}, 2023.

\bibitem[Aggarwal et~al.(2022)Aggarwal, Pramanik, John, and Jacob]{aggarwal2022ensure}
Hemant~Kumar Aggarwal, Aniket Pramanik, Maneesh John, and Mathews Jacob.
\newblock {ENSURE}: A general approach for unsupervised training of deep image reconstruction algorithms.
\newblock \emph{IEEE Transactions on Medical Imaging}, 2022.

\bibitem[Bao et~al.(2022)Bao, Li, Sun, Zhu, and Zhang]{bao2022estimating}
Fan Bao, Chongxuan Li, Jiacheng Sun, Jun Zhu, and Bo~Zhang.
\newblock Estimating the optimal covariance with imperfect mean in diffusion probabilistic models.
\newblock In \emph{International Conference on Machine Learning}, pp.\  1555--1584. PMLR, 2022.

\bibitem[Batson \& Royer(2019)Batson and Royer]{batson2019noise2self}
Joshua Batson and Loic Royer.
\newblock Noise2self: Blind denoising by self-supervision.
\newblock In \emph{International Conference on Machine Learning}, pp.\  524--533. PMLR, 2019.

\bibitem[Blau et~al.(2022)Blau, Ganz, Kawar, Bronstein, and Elad]{blau2022threat}
Tsachi Blau, Roy Ganz, Bahjat Kawar, Alex Bronstein, and Michael Elad.
\newblock Threat model-agnostic adversarial defense using diffusion models.
\newblock \emph{arXiv preprint arXiv:2207.08089}, 2022.

\bibitem[Blu \& Luisier(2007)Blu and Luisier]{blu2007sure}
Thierry Blu and Florian Luisier.
\newblock The sure-let approach to image denoising.
\newblock \emph{IEEE Transactions on Image Processing}, 16\penalty0 (11):\penalty0 2778--2786, 2007.

\bibitem[Chen et~al.(2022)Chen, Tachella, and Davies]{chen2022robust}
Dongdong Chen, Juli{\'a}n Tachella, and Mike~E Davies.
\newblock Robust equivariant imaging: a fully unsupervised framework for learning to image from noisy and partial measurements.
\newblock In \emph{Proceedings of the IEEE/CVF Conference on Computer Vision and Pattern Recognition}, pp.\  5647--5656, 2022.

\bibitem[Cheng-Xian~Li et~al.(2019)Cheng-Xian~Li, Jiang, and Marlin]{cheng2019misgan}
Steven Cheng-Xian~Li, Bo~Jiang, and Benjamin Marlin.
\newblock Misgan: Learning from incomplete data with generative adversarial networks.
\newblock \emph{arXiv e-prints}, pp.\  arXiv--1902, 2019.

\bibitem[Chung \& Ye(2022)Chung and Ye]{chung2022score}
Hyungjin Chung and Jong~Chul Ye.
\newblock Score-based diffusion models for accelerated {MRI}.
\newblock \emph{Medical Image Analysis}, 80:\penalty0 102479, 2022.

\bibitem[Chung et~al.(2023)Chung, Ryu, McCann, Klasky, and Ye]{chung2023solving}
Hyungjin Chung, Dohoon Ryu, Michael~T McCann, Marc~L Klasky, and Jong~Chul Ye.
\newblock Solving {3D} inverse problems using pre-trained {2D} diffusion models.
\newblock In \emph{IEEE/CVF Conference on Computer Vision and Pattern Recognition}, 2023.

\bibitem[Daras et~al.(2024)Daras, Shah, Dagan, Gollakota, Dimakis, and Klivans]{daras2024ambient}
Giannis Daras, Kulin Shah, Yuval Dagan, Aravind Gollakota, Alex Dimakis, and Adam Klivans.
\newblock Ambient diffusion: Learning clean distributions from corrupted data.
\newblock \emph{Advances in Neural Information Processing Systems}, 36, 2024.

\bibitem[Deng et~al.(2009)Deng, Dong, Socher, Li, Li, and Fei-Fei]{imagenet}
Jia Deng, Wei Dong, Richard Socher, Li-Jia Li, Kai Li, and Li~Fei-Fei.
\newblock {ImageNet: A large-scale hierarchical image database}.
\newblock In \emph{2009 IEEE Conference on Computer Vision and Pattern Recognition}, pp.\  248--255, 2009.

\bibitem[Dhariwal \& Nichol(2021)Dhariwal and Nichol]{dhariwal2021diffusion}
Prafulla Dhariwal and Alexander Nichol.
\newblock Diffusion models beat {GANs} on image synthesis.
\newblock \emph{Advances in Neural Information Processing Systems}, 34:\penalty0 8780--8794, 2021.

\bibitem[Eldar(2008)]{eldar2008generalized}
Yonina~C Eldar.
\newblock Generalized {SURE} for exponential families: Applications to regularization.
\newblock \emph{IEEE Transactions on Signal Processing}, 57\penalty0 (2):\penalty0 471--481, 2008.

\bibitem[Gao et~al.(2023)Gao, Zhou, Cheng, and Yan]{gao2023masked}
Shanghua Gao, Pan Zhou, Ming-Ming Cheng, and Shuicheng Yan.
\newblock Masked diffusion transformer is a strong image synthesizer.
\newblock \emph{arXiv preprint arXiv:2303.14389}, 2023.

\bibitem[Gong et~al.(2023)Gong, Li, Feng, Wu, and Kong]{gong2022diffuseq}
Shansan Gong, Mukai Li, Jiangtao Feng, Zhiyong Wu, and Lingpeng Kong.
\newblock {DiffuSeq}: Sequence to sequence text generation with diffusion models.
\newblock In \emph{International Conference on Learning Representations}, 2023.

\bibitem[Hammernik et~al.(2018)Hammernik, Klatzer, Kobler, Recht, Sodickson, Pock, and Knoll]{hammernik2018learning}
Kerstin Hammernik, Teresa Klatzer, Erich Kobler, Michael~P Recht, Daniel~K Sodickson, Thomas Pock, and Florian Knoll.
\newblock Learning a variational network for reconstruction of accelerated mri data.
\newblock \emph{Magnetic resonance in medicine}, 79\penalty0 (6):\penalty0 3055--3071, 2018.

\bibitem[Han et~al.(2019)Han, Sunwoo, and Ye]{han2019k}
Yoseo Han, Leonard Sunwoo, and Jong~Chul Ye.
\newblock $\{$k$\}$-space deep learning for accelerated mri.
\newblock \emph{IEEE transactions on medical imaging}, 39\penalty0 (2):\penalty0 377--386, 2019.

\bibitem[Hendriksen et~al.(2020)Hendriksen, Pelt, and Batenburg]{hendriksen2020noise2inverse}
Allard~Adriaan Hendriksen, Dani{\"e}l~Maria Pelt, and K~Joost Batenburg.
\newblock Noise2inverse: Self-supervised deep convolutional denoising for tomography.
\newblock \emph{IEEE Transactions on Computational Imaging}, 6:\penalty0 1320--1335, 2020.

\bibitem[Heusel et~al.(2017)Heusel, Ramsauer, Unterthiner, Nessler, and Hochreiter]{fid}
Martin Heusel, Hubert Ramsauer, Thomas Unterthiner, Bernhard Nessler, and Sepp Hochreiter.
\newblock Gans trained by a two time-scale update rule converge to a local nash equilibrium.
\newblock In \emph{Advances in Neural Information Processing Systems}, volume~30, 2017.

\bibitem[Ho et~al.(2020)Ho, Jain, and Abbeel]{ho2020denoising}
Jonathan Ho, Ajay Jain, and Pieter Abbeel.
\newblock Denoising diffusion probabilistic models.
\newblock \emph{Advances in Neural Information Processing Systems}, 33:\penalty0 6840--6851, 2020.

\bibitem[Ho et~al.(2022)Ho, Chan, Saharia, Whang, Gao, Gritsenko, Kingma, Poole, Norouzi, Fleet, et~al.]{ho2022imagen}
Jonathan Ho, William Chan, Chitwan Saharia, Jay Whang, Ruiqi Gao, Alexey Gritsenko, Diederik~P Kingma, Ben Poole, Mohammad Norouzi, David~J Fleet, et~al.
\newblock Imagen video: High definition video generation with diffusion models.
\newblock \emph{arXiv preprint arXiv:2210.02303}, 2022.

\bibitem[Hsieh et~al.(2023)Hsieh, Kasiviswanathan, Kveton, and Bl{\"o}baum]{hsieh2023thompson}
Yu-Guan Hsieh, Shiva~Prasad Kasiviswanathan, Branislav Kveton, and Patrick Bl{\"o}baum.
\newblock Thompson sampling with diffusion generative prior.
\newblock \emph{arXiv preprint arXiv:2301.05182}, 2023.

\bibitem[Hutchinson(1989)]{hutchinson1989stochastic}
Michael~F Hutchinson.
\newblock A stochastic estimator of the trace of the influence matrix for laplacian smoothing splines.
\newblock \emph{Communications in Statistics-Simulation and Computation}, 18\penalty0 (3):\penalty0 1059--1076, 1989.

\bibitem[Jalal et~al.(2021)Jalal, Arvinte, Daras, Price, Dimakis, and Tamir]{jalal2021robust}
Ajil Jalal, Marius Arvinte, Giannis Daras, Eric Price, Alexandros~G Dimakis, and Jon Tamir.
\newblock Robust compressed sensing mri with deep generative priors.
\newblock \emph{Advances in Neural Information Processing Systems}, 34:\penalty0 14938--14954, 2021.

\bibitem[Jo et~al.(2021)Jo, Chun, and Choi]{jo2021rethinking}
Yeonsik Jo, Se~Young Chun, and Jonghyun Choi.
\newblock Rethinking deep image prior for denoising.
\newblock In \emph{Proceedings of the IEEE/CVF International Conference on Computer Vision}, pp.\  5087--5096, 2021.

\bibitem[Kawar et~al.(2022)Kawar, Elad, Ermon, and Song]{kawar2022denoising}
Bahjat Kawar, Michael Elad, Stefano Ermon, and Jiaming Song.
\newblock Denoising diffusion restoration models.
\newblock In \emph{Advances in Neural Information Processing Systems}, 2022.

\bibitem[Kawar et~al.(2023)Kawar, Zada, Lang, Tov, Chang, Dekel, Mosseri, and Irani]{kawar2023imagic}
Bahjat Kawar, Shiran Zada, Oran Lang, Omer Tov, Huiwen Chang, Tali Dekel, Inbar Mosseri, and Michal Irani.
\newblock Imagic: Text-based real image editing with diffusion models.
\newblock In \emph{Conference on Computer Vision and Pattern Recognition 2023}, 2023.

\bibitem[Kim et~al.(2022)Kim, Kim, Kang, and Moon]{kim2022refining}
Dongjun Kim, Yeongmin Kim, Wanmo Kang, and Il-Chul Moon.
\newblock Refining generative process with discriminator guidance in score-based diffusion models.
\newblock \emph{arXiv preprint arXiv:2211.17091}, 2022.

\bibitem[Knoll et~al.(2020)Knoll, Zbontar, Sriram, Muckley, Bruno, Defazio, Parente, Geras, Katsnelson, Chandarana, Zhang, Drozdzalv, Romero, Rabbat, Vincent, Pinkerton, Wang, Yakubova, Owens, Zitnick, Recht, Sodickson, and Lui]{knoll2020fastmri}
Florian Knoll, Jure Zbontar, Anuroop Sriram, Matthew~J. Muckley, Mary Bruno, Aaron Defazio, Marc Parente, Krzysztof~J. Geras, Joe Katsnelson, Hersh Chandarana, Zizhao Zhang, Michal Drozdzalv, Adriana Romero, Michael Rabbat, Pascal Vincent, James Pinkerton, Duo Wang, Nafissa Yakubova, Erich Owens, C.~Lawrence Zitnick, Michael~P. Recht, Daniel~K. Sodickson, and Yvonne~W. Lui.
\newblock {fastMRI}: A publicly available raw k-space and {DICOM} dataset of knee images for accelerated {MR} image reconstruction using machine learning.
\newblock \emph{Radiology: Artificial Intelligence}, 2\penalty0 (1):\penalty0 e190007, 2020.

\bibitem[Kong et~al.(2021)Kong, Ping, Huang, Zhao, and Catanzaro]{kong2021diffwave}
Zhifeng Kong, Wei Ping, Jiaji Huang, Kexin Zhao, and Bryan Catanzaro.
\newblock Diffwave: A versatile diffusion model for audio synthesis.
\newblock In \emph{International Conference on Learning Representations}, 2021.

\bibitem[Lee et~al.(2018)Lee, Yoo, Tak, and Ye]{lee2018deep}
Dongwook Lee, Jaejun Yoo, Sungho Tak, and Jong~Chul Ye.
\newblock Deep residual learning for accelerated mri using magnitude and phase networks.
\newblock \emph{IEEE Transactions on Biomedical Engineering}, 65\penalty0 (9):\penalty0 1985--1995, 2018.

\bibitem[Lehtinen et~al.(2018)Lehtinen, Munkberg, Hasselgren, Laine, Karras, Aittala, and Aila]{lehtinen2018noise2noise}
Jaakko Lehtinen, Jacob Munkberg, Jon Hasselgren, Samuli Laine, Tero Karras, Miika Aittala, and Timo Aila.
\newblock Noise2noise: Learning image restoration without clean data.
\newblock In \emph{International Conference on Machine Learning}, pp.\  2965--2974. PMLR, 2018.

\bibitem[Li et~al.(2023)Li, Ditzler, Roveda, and Li]{li2023descod}
Huayu Li, Gregory Ditzler, Janet Roveda, and Ao~Li.
\newblock {DeScoD-ECG}: Deep score-based diffusion model for ecg baseline wander and noise removal.
\newblock \emph{IEEE Journal of Biomedical and Health Informatics}, 2023.

\bibitem[Li et~al.(2022)Li, Thickstun, Gulrajani, Liang, and Hashimoto]{li2022diffusionlm}
Xiang~Lisa Li, John Thickstun, Ishaan Gulrajani, Percy Liang, and Tatsunori Hashimoto.
\newblock Diffusion-{LM} improves controllable text generation.
\newblock In \emph{Advances in Neural Information Processing Systems}, 2022.

\bibitem[Liu et~al.(2020)Liu, Sun, Eldeniz, Gan, An, and Kamilov]{liu2020rare}
Jiaming Liu, Yu~Sun, Cihat Eldeniz, Weijie Gan, Hongyu An, and Ulugbek~S Kamilov.
\newblock {RARE}: Image reconstruction using deep priors learned without groundtruth.
\newblock \emph{IEEE Journal of Selected Topics in Signal Processing}, 14\penalty0 (6):\penalty0 1088--1099, 2020.

\bibitem[Liu et~al.(2015)Liu, Luo, Wang, and Tang]{liu2015celeba}
Ziwei Liu, Ping Luo, Xiaogang Wang, and Xiaoou Tang.
\newblock Deep learning face attributes in the wild.
\newblock In \emph{Proceedings of the IEEE International Conference on Computer Vision}, pp.\  3730--3738, 2015.

\bibitem[Mattei \& Frellsen(2019)Mattei and Frellsen]{mattei2019miwae}
Pierre-Alexandre Mattei and Jes Frellsen.
\newblock Miwae: Deep generative modelling and imputation of incomplete data sets.
\newblock In \emph{International conference on machine learning}, pp.\  4413--4423. PMLR, 2019.

\bibitem[Metzler et~al.(2018)Metzler, Mousavi, Heckel, and Baraniuk]{metzler2018unsupervised}
Christopher~A Metzler, Ali Mousavi, Reinhard Heckel, and Richard~G Baraniuk.
\newblock Unsupervised learning with stein's unbiased risk estimator.
\newblock \emph{arXiv preprint arXiv:1805.10531}, 2018.

\bibitem[Mullainathan \& Obermeyer(2022)Mullainathan and Obermeyer]{Mullainathan2022}
Sendhil Mullainathan and Ziad Obermeyer.
\newblock Solving medicine's data bottleneck: Nightingale open science.
\newblock \emph{Nature Medicine}, 28\penalty0 (5):\penalty0 897--899, May 2022.

\bibitem[Nguyen et~al.(2020)Nguyen, Ulfarsson, and Sveinsson]{nguyen2020hyperspectral}
Han~V Nguyen, Magnus~O Ulfarsson, and Johannes~R Sveinsson.
\newblock Hyperspectral image denoising using sure-based unsupervised convolutional neural networks.
\newblock \emph{IEEE Transactions on Geoscience and Remote Sensing}, 59\penalty0 (4):\penalty0 3369--3382, 2020.

\bibitem[Nichol \& Dhariwal(2021)Nichol and Dhariwal]{nichol2021improved}
Alexander~Quinn Nichol and Prafulla Dhariwal.
\newblock Improved denoising diffusion probabilistic models.
\newblock In \emph{International Conference on Machine Learning}, pp.\  8162--8171. PMLR, 2021.

\bibitem[Pinaya et~al.(2022)Pinaya, Graham, Gray, Da~Costa, Tudosiu, Wright, Mah, MacKinnon, Teo, Jager, et~al.]{pinaya2022fast}
Walter~HL Pinaya, Mark~S Graham, Robert Gray, Pedro~F Da~Costa, Petru-Daniel Tudosiu, Paul Wright, Yee~H Mah, Andrew~D MacKinnon, James~T Teo, Rolf Jager, et~al.
\newblock Fast unsupervised brain anomaly detection and segmentation with diffusion models.
\newblock In \emph{Medical Image Computing and Computer Assisted Intervention--MICCAI 2022: 25th International Conference, Singapore, September 18--22, 2022, Proceedings, Part VIII}, pp.\  705--714. Springer, 2022.

\bibitem[Popov et~al.(2021)Popov, Vovk, Gogoryan, Sadekova, and Kudinov]{popov2021grad}
Vadim Popov, Ivan Vovk, Vladimir Gogoryan, Tasnima Sadekova, and Mikhail Kudinov.
\newblock {Grad-TTS}: A diffusion probabilistic model for text-to-speech.
\newblock In \emph{International Conference on Machine Learning}, pp.\  8599--8608. PMLR, 2021.

\bibitem[Qiao et~al.(2022)Qiao, Nie, Vahdat, Miller~III, and Anandkumar]{qiao2022dynamic}
Zhuoran Qiao, Weili Nie, Arash Vahdat, Thomas~F Miller~III, and Anima Anandkumar.
\newblock Dynamic-backbone protein-ligand structure prediction with multiscale generative diffusion models.
\newblock \emph{arXiv preprint arXiv:2209.15171}, 2022.

\bibitem[Ramani et~al.(2008)Ramani, Blu, and Unser]{ramani2008monte}
Sathish Ramani, Thierry Blu, and Michael Unser.
\newblock Monte-carlo sure: A black-box optimization of regularization parameters for general denoising algorithms.
\newblock \emph{IEEE Transactions on image processing}, 17\penalty0 (9):\penalty0 1540--1554, 2008.

\bibitem[Rombach et~al.(2022)Rombach, Blattmann, Lorenz, Esser, and Ommer]{latent_diffusion}
Robin Rombach, Andreas Blattmann, Dominik Lorenz, Patrick Esser, and Bj{\"o}rn Ommer.
\newblock High-resolution image synthesis with latent diffusion models.
\newblock In \emph{Proceedings of the IEEE/CVF Conference on Computer Vision and Pattern Recognition}, pp.\  10684--10695, 2022.

\bibitem[Ronneberger et~al.(2015)Ronneberger, Fischer, and Brox]{unet}
Olaf Ronneberger, Philipp Fischer, and Thomas Brox.
\newblock {U-Net}: Convolutional networks for biomedical image segmentation.
\newblock In \emph{Medical Image Computing and Computer-Assisted Intervention--MICCAI 2015: 18th International Conference, Munich, Germany, October 5-9, 2015, Proceedings, Part III 18}, pp.\  234--241. Springer, 2015.

\bibitem[Schneuing et~al.(2022)Schneuing, Du, Harris, Jamasb, Igashov, Du, Blundell, Li{\'o}, Gomes, Welling, Bronstein, and Correia]{schneuing2022structure}
Arne Schneuing, Yuanqi Du, Charles Harris, Arian Jamasb, Ilia Igashov, Weitao Du, Tom Blundell, Pietro Li{\'o}, Carla Gomes, Max Welling, Michael Bronstein, and Bruno Correia.
\newblock Structure-based drug design with equivariant diffusion models.
\newblock \emph{arXiv preprint arXiv:2210.13695}, 2022.

\bibitem[Schuhmann et~al.(2022)Schuhmann, Beaumont, Vencu, Gordon, Wightman, Cherti, Coombes, Katta, Mullis, Wortsman, et~al.]{schuhmann2022laion}
Christoph Schuhmann, Romain Beaumont, Richard Vencu, Cade~W Gordon, Ross Wightman, Mehdi Cherti, Theo Coombes, Aarush Katta, Clayton Mullis, Mitchell Wortsman, et~al.
\newblock Laion-5b: An open large-scale dataset for training next generation image-text models.
\newblock In \emph{Thirty-sixth Conference on Neural Information Processing Systems Datasets and Benchmarks Track}, 2022.

\bibitem[Singer et~al.(2023)Singer, Polyak, Hayes, Yin, An, Zhang, Hu, Yang, Ashual, Gafni, Parikh, Gupta, and Taigman]{singer2023makeavideo}
Uriel Singer, Adam Polyak, Thomas Hayes, Xi~Yin, Jie An, Songyang Zhang, Qiyuan Hu, Harry Yang, Oron Ashual, Oran Gafni, Devi Parikh, Sonal Gupta, and Yaniv Taigman.
\newblock {Make-A-Video}: Text-to-video generation without text-video data.
\newblock In \emph{International Conference on Learning Representations}, 2023.

\bibitem[Sohl-Dickstein et~al.(2015)Sohl-Dickstein, Weiss, Maheswaranathan, and Ganguli]{sohl2015deep}
Jascha Sohl-Dickstein, Eric Weiss, Niru Maheswaranathan, and Surya Ganguli.
\newblock Deep unsupervised learning using nonequilibrium thermodynamics.
\newblock In \emph{International Conference on Machine Learning}, pp.\  2256--2265. PMLR, 2015.

\bibitem[Soltanayev \& Chun(2018)Soltanayev and Chun]{soltanayev2018training}
Shakarim Soltanayev and Se~Young Chun.
\newblock Training deep learning based denoisers without ground truth data.
\newblock \emph{Advances in neural information processing systems}, 31, 2018.

\bibitem[Song et~al.(2020{\natexlab{a}})Song, Meng, and Ermon]{song2020denoising}
Jiaming Song, Chenlin Meng, and Stefano Ermon.
\newblock Denoising diffusion implicit models.
\newblock In \emph{International Conference on Learning Representations}, 2020{\natexlab{a}}.

\bibitem[Song \& Ermon(2019)Song and Ermon]{song2019generative}
Yang Song and Stefano Ermon.
\newblock Generative modeling by estimating gradients of the data distribution.
\newblock \emph{Advances in Neural Information Processing Systems}, 32, 2019.

\bibitem[Song et~al.(2020{\natexlab{b}})Song, Sohl-Dickstein, Kingma, Kumar, Ermon, and Poole]{song2020score}
Yang Song, Jascha Sohl-Dickstein, Diederik~P Kingma, Abhishek Kumar, Stefano Ermon, and Ben Poole.
\newblock Score-based generative modeling through stochastic differential equations.
\newblock In \emph{International Conference on Learning Representations}, 2020{\natexlab{b}}.

\bibitem[Song et~al.(2023)Song, Shen, Xing, and Ermon]{song2023solving}
Yang Song, Liyue Shen, Lei Xing, and Stefano Ermon.
\newblock Solving inverse problems in medical imaging with score-based generative models.
\newblock In \emph{International Conference on Learning Representations}, 2023.

\bibitem[Stein(1981)]{stein1981estimation}
Charles~M Stein.
\newblock Estimation of the mean of a multivariate normal distribution.
\newblock \emph{The annals of Statistics}, pp.\  1135--1151, 1981.

\bibitem[Tevet et~al.(2022)Tevet, Raab, Gordon, Shafir, Bermano, and Cohen-Or]{tevet2022human}
Guy Tevet, Sigal Raab, Brian Gordon, Yonatan Shafir, Amit~H Bermano, and Daniel Cohen-Or.
\newblock Human motion diffusion model.
\newblock \emph{arXiv preprint arXiv:2209.14916}, 2022.

\bibitem[Theis et~al.(2022)Theis, Salimans, Hoffman, and Mentzer]{theis2022lossy}
Lucas Theis, Tim Salimans, Matthew~D Hoffman, and Fabian Mentzer.
\newblock Lossy compression with {G}aussian diffusion.
\newblock \emph{arXiv preprint arXiv:2206.08889}, 2022.

\bibitem[Vahdat et~al.(2021)Vahdat, Kreis, and Kautz]{vahdat2021score}
Arash Vahdat, Karsten Kreis, and Jan Kautz.
\newblock Score-based generative modeling in latent space.
\newblock \emph{Advances in Neural Information Processing Systems}, 34:\penalty0 11287--11302, 2021.

\bibitem[Wang et~al.(2016)Wang, Su, Ying, Peng, Zhu, Liang, Feng, and Liang]{wang2016accelerating}
Shanshan Wang, Zhenghang Su, Leslie Ying, Xi~Peng, Shun Zhu, Feng Liang, Dagan Feng, and Dong Liang.
\newblock Accelerating magnetic resonance imaging via deep learning.
\newblock In \emph{IEEE 13th International Symposium on Biomedical Imaging (ISBI)}, pp.\  514--517. IEEE, 2016.

\bibitem[Wang et~al.(2022)Wang, Qian, Guo, Sun, Li, Zhao, and Qu]{wang2022one}
Zi~Wang, Chen Qian, Di~Guo, Hongwei Sun, Rushuai Li, Bo~Zhao, and Xiaobo Qu.
\newblock One-dimensional deep low-rank and sparse network for accelerated mri.
\newblock \emph{IEEE Transactions on Medical Imaging}, 42\penalty0 (1):\penalty0 79--90, 2022.

\bibitem[Watson et~al.(2022)Watson, Juergens, Bennett, Trippe, Yim, Eisenach, Ahern, Borst, Ragotte, Milles, et~al.]{watson2022broadly}
Joseph~L Watson, David Juergens, Nathaniel~R Bennett, Brian~L Trippe, Jason Yim, Helen~E Eisenach, Woody Ahern, Andrew~J Borst, Robert~J Ragotte, Lukas~F Milles, et~al.
\newblock Broadly applicable and accurate protein design by integrating structure prediction networks and diffusion generative models.
\newblock \emph{bioRxiv}, pp.\  2022--12, 2022.

\bibitem[Weiss et~al.(2021)Weiss, Senouf, Vedula, Michailovich, Zibulevsky, Bronstein, et~al.]{weiss2021pilot}
Tomer Weiss, Ortal Senouf, Sanketh Vedula, Oleg Michailovich, Michael Zibulevsky, Alex Bronstein, et~al.
\newblock Pilot: Physics-informed learned optimized trajectories for accelerated mri.
\newblock \emph{Machine Learning for Biomedical Imaging}, 1\penalty0 (April 2021 issue):\penalty0 1--23, 2021.

\bibitem[Wyatt et~al.(2022)Wyatt, Leach, Schmon, and Willcocks]{wyatt2022anoddpm}
Julian Wyatt, Adam Leach, Sebastian~M Schmon, and Chris~G Willcocks.
\newblock {AnoDDPM}: Anomaly detection with denoising diffusion probabilistic models using simplex noise.
\newblock In \emph{Proceedings of the IEEE/CVF Conference on Computer Vision and Pattern Recognition}, pp.\  650--656, 2022.

\bibitem[Xiang et~al.(2023)Xiang, Yurt, Syed, Setsompop, and Chaudhari]{xiangddm}
Tiange Xiang, Mahmut Yurt, Ali~B Syed, Kawin Setsompop, and Akshay Chaudhari.
\newblock {DDM$^{2}$}: Self-supervised diffusion {MRI} denoising with generative diffusion models.
\newblock In \emph{The Eleventh International Conference on Learning Representations}, 2023.

\bibitem[Xie \& Li(2022)Xie and Li]{xie2022measurement}
Yutong Xie and Quanzheng Li.
\newblock Measurement-conditioned denoising diffusion probabilistic model for under-sampled medical image reconstruction.
\newblock In \emph{Medical Image Computing and Computer Assisted Intervention--MICCAI 2022: 25th International Conference, Singapore, September 18--22, 2022, Proceedings, Part VI}, pp.\  655--664. Springer, 2022.

\bibitem[Yim et~al.(2023)Yim, Trippe, De~Bortoli, Mathieu, Doucet, Barzilay, and Jaakkola]{yim2023se}
Jason Yim, Brian~L Trippe, Valentin De~Bortoli, Emile Mathieu, Arnaud Doucet, Regina Barzilay, and Tommi Jaakkola.
\newblock Se(3) diffusion model with application to protein backbone generation.
\newblock \emph{arXiv preprint arXiv:2302.02277}, 2023.

\bibitem[Zbontar et~al.(2019)Zbontar, Knoll, Sriram, Murrell, Huang, Muckley, Defazio, Stern, Johnson, Bruno, Parente, Geras, Katsnelson, Chandarana, Zhang, Drozdzal, Romero, Rabbat, Vincent, Yakubova, Pinkerton, Wang, Owens, Zitnick, Recht, Sodickson, and Lui]{zbontar2019fastmri}
Jure Zbontar, Florian Knoll, Anuroop Sriram, Tullie Murrell, Zhengnan Huang, Matthew~J. Muckley, Aaron Defazio, Ruben Stern, Patricia Johnson, Mary Bruno, Marc Parente, Krzysztof~J. Geras, Joe Katsnelson, Hersh Chandarana, Zizhao Zhang, Michal Drozdzal, Adriana Romero, Michael Rabbat, Pascal Vincent, Nafissa Yakubova, James Pinkerton, Duo Wang, Erich Owens, C.~Lawrence Zitnick, Michael~P. Recht, Daniel~K. Sodickson, and Yvonne~W. Lui.
\newblock {fastMRI}: An open dataset and benchmarks for accelerated {MRI}, 2019.

\bibitem[Zhang \& Desai(1998)Zhang and Desai]{zhang1998adaptive}
Xiao-Ping Zhang and Mita~D Desai.
\newblock Adaptive denoising based on sure risk.
\newblock \emph{IEEE signal processing letters}, 5\penalty0 (10):\penalty0 265--267, 1998.

\bibitem[Zhussip et~al.(2019)Zhussip, Soltanayev, and Chun]{zhussip2019training}
Magauiya Zhussip, Shakarim Soltanayev, and Se~Young Chun.
\newblock Training deep learning based image denoisers from undersampled measurements without ground truth and without image prior.
\newblock In \emph{Proceedings of the IEEE/CVF Conference on Computer Vision and Pattern Recognition}, pp.\  10255--10264, 2019.

\bibitem[Zimmermann et~al.(2021)Zimmermann, Schott, Song, Dunn, and Klindt]{zimmermann2021score}
Roland~S Zimmermann, Lukas Schott, Yang Song, Benjamin~A Dunn, and David~A Klindt.
\newblock Score-based generative classifiers.
\newblock \emph{arXiv preprint arXiv:2110.00473}, 2021.

\end{thebibliography}
\bibliographystyle{tmlr}

\newpage
\appendix
\section{Proposition Proofs}
\label{sec:proofs}

\ensemble*
\begin{proof}
We focus on the expectation term from \autoref{eq:projected-diffusion-loss}, which is taken over ${\xbar_t \sim q(\xbar_t | \xbar, \mP)},\ {\xbar \sim q(\xbar)},\ {\mP \sim q_P(\mP)}$ with unknown $q(\xbar), q_P(\mP)$. Namely,
\begin{align}
    & \bbE \left[ \left\lVert \mW \mP \left(  f_\theta^{(t)}(\xbar_t) - \xbar \right) \right\rVert_2^2 \right] \nonumber \\
    \stackrel{1}{=}\  & \bbE \left[ \mathrm{Trace} \left(
    \mW \mP \left(  f_\theta^{(t)}(\xbar_t) - \xbar \right)
    \left(  f_\theta^{(t)}(\xbar_t) - \xbar \right)^\top \mP \mW
    \right) \right] \nonumber \\
    \stackrel{2}{=}\  & \bbE \left[ \mathrm{Trace} \left(
    \mP \mW^2 \mP \left(  f_\theta^{(t)}(\xbar_t) - \xbar \right)
    \left(  f_\theta^{(t)}(\xbar_t) - \xbar \right)^\top
    \right) \right] \nonumber \\
    \stackrel{3}{=}\  & \bbE \left[ \mathrm{Trace} \left(
    \mW^2 \mP^2 \left(  f_\theta^{(t)}(\xbar_t) - \xbar \right)
    \left(  f_\theta^{(t)}(\xbar_t) - \xbar \right)^\top
    \right) \right] \nonumber \\
    \stackrel{4}{=}\  & \bbE \left[ \mathrm{Trace} \left(
    \mW^2 \mP \left(  f_\theta^{(t)}(\xbar_t) - \xbar \right)
    \left(  f_\theta^{(t)}(\xbar_t) - \xbar \right)^\top
    \right) \right] \nonumber \\
    \stackrel{5}{=}\  & \mathrm{Trace} \left( \bbE \left[
    \mW^2 \mP \left(  f_\theta^{(t)}(\xbar_t) - \xbar \right)
    \left(  f_\theta^{(t)}(\xbar_t) - \xbar \right)^\top
    \right] \right) \nonumber \\
    \stackrel{6}{=}\  & \mathrm{Trace} \left( \mW^2 \bbE\left[\mP\right] \bbE \left[
    \left(  f_\theta^{(t)}(\xbar_t) - \xbar \right)
    \left(  f_\theta^{(t)}(\xbar_t) - \xbar \right)^\top
    \right] \right) \nonumber \\
    \stackrel{7}{=}\  & \mathrm{Trace} \left( \bbE \left[
    \left(  f_\theta^{(t)}(\xbar_t) - \xbar \right)
    \left(  f_\theta^{(t)}(\xbar_t) - \xbar \right)^\top
    \right] \right) \nonumber \\
    \stackrel{5}{=}\  & \bbE \left[ \mathrm{Trace} \left(
    \left(  f_\theta^{(t)}(\xbar_t) - \xbar \right)
    \left(  f_\theta^{(t)}(\xbar_t) - \xbar \right)^\top
    \right) \right] \nonumber \\
    \stackrel{8}{=}\  & \bbE \left[ \left\lVert
    f_\theta^{(t)}(\xbar_t) - \xbar \right\rVert_2^2\right].
    \label{eq:proof-diff-loss}
\end{align}
Justifications:
\begin{enumerate}
    \item Using the linear algebra property $\lVert\rvv\rVert_2^2 = \mathrm{Trace}\left(\rvv \rvv^\top\right)$ for any vector $\rvv$, and ${\mP = \mP^\top}, {\mW = \mW^\top}$ because they are diagonal matrices.
    \item Using the cyclical shift invariance of the trace operator, $\mathrm{Trace}\left(\mA \mB \mC\right) = \mathrm{Trace}\left(\mC \mA \mB\right)$.
    \item Diagonal matrices (such as $\mP$ and $\mW$) commute with one another.
    \item Since $\mP$ is a diagonal matrix whose values are either zeroes or ones, it holds $\mP^2 = \mP$.
    \item For any random matrix $\mA$, it holds that $\bbE\left[\mathrm{Trace}\left(\mA\right)\right] = \mathrm{Trace}\left(\bbE\left[\mA\right]\right)$. 
    \item $\mW^2$ is a constant and can therefore be taken out of the expectation.
    Additionally, we use the assumption that the denoiser's error $\left( f_\theta^{(t)}(\xbar_t) - \xbar \right)$ is independent of $\mP$.
    \item $\mW$ is defined as $\bbE[\mP]^{-\frac{1}{2}}$. This results in $\mW^2 \mP = \mI$.
    \item Using the linear algebra property $\lVert\rvv\rVert_2^2 = \mathrm{Trace}\left(\rvv \rvv^\top\right)$ for any vector $\rvv$.
\end{enumerate}
\autoref{eq:proof-diff-loss} is identical to the expectation term in \autoref{eq:diffusion-loss}, except for the distribution of $\xbar_t$ considered in the expectation.
\autoref{eq:proof-diff-loss} considers ${\xbar_t \sim q(\xbar_t | \xbar, \mP) = \gN(\sqrt{\bar{\alpha}_t} \mP \xbar, (1 - \bar{\alpha}_t) \mI})$, whereas \autoref{eq:diffusion-loss} considers ${\xbar_t \sim q^*(\xbar_t | \xbar) = \gN(\sqrt{\bar{\alpha}_t} \xbar, (1 - \bar{\alpha}_t) \mI})$.
We assume that the neural network $f_\theta^{(t)}(\xbar_t)$ is able to infer $\mP$ from $\xbar_t$, and can also tailor its output for each $\mP$ including $\mP = \mI$, matching $q^*(\xbar_t | \xbar)$.
Under these assumptions, both expectations share the same minimizer, thereby completing the proof.
A similar proof is presented in ENSURE~\cite{aggarwal2022ensure}.
\end{proof}

\unbiased*
\begin{proof}
We utilize a weighted version of the generalized SURE~\cite{eldar2008generalized, aggarwal2022ensure} presented in \autoref{eq:gsure},
\begin{equation}
\label{eq:weighted-gsure}
    \bbE\left[ \left\lVert \mW \mP \left(f(\rvy) - \rvx \right) \right\rVert_2^2 \right] =
    \bbE\left[ \left\lVert \mW \mP \left(f(\rvy) - \rvx_\mathrm{ML} \right) \right\rVert_2^2 \right]
    + 2 \bbE\left[ \nabla_{\mH^\top \mC^{-1} \rvy} \cdot \mW^2 \mP f(\rvy) \right]
    + c.
\end{equation}
This weighted GSURE considers the measurement equation $\rvy = \mH \rvx + \rvz$ with $\rvz \sim \gN(0, \mC)$, with $\mP = \mH^\dagger \mH$, $\rvx_\mathrm{ML} = \left(\mH^\top \mC^{-1} \mH\right)^\dagger \mH^\top \mC^{-1} \rvy$, and a constant $c$.
We consider the measurement equation matching \autoref{eq:xbart-marginal}, namely $\xbar_t = \sqrt{\bar{\alpha}_t}\mP \xbar + \zbar_t$ with ${\zbar_t \sim \gN(0, (1 - \bar{\alpha}_t) \mI)}$, which is a special case.
For these measurements, the left-hand-side in \autoref{eq:weighted-gsure} becomes
\begin{align*}
    \bbE\left[ \left\lVert \mW \left(\sqrt{\bar{\alpha}_t} \mP\right)^\dagger \left(\sqrt{\bar{\alpha}_t} \mP\right) \left(f(\xbar_t) - \xbar \right) \right\rVert_2^2 \right] = \bbE\left[ \left\lVert \mW \mP \left(f(\xbar_t) - \xbar \right) \right\rVert_2^2 \right],
\end{align*}
which is identical to the expectation term in \autoref{eq:projected-diffusion-loss}.
This equation holds because $\mP$ is a diagonal matrix with ones and zeroes, resulting in ${\mP^\dagger = \mP = \mP^2}$.
Meanwhile, by substituting $\mH = \sqrt{\bar{\alpha}_t} \mP$, $\mC = (1 - \bar{\alpha}_t) \mI$, and $\rvy = \xbar_t$,
$\rvx_\mathrm{ML}$ simplifies into
\begin{align*}
    \rvx_\mathrm{ML} & = \left(\sqrt{\bar{\alpha}_t} \mP^\top \left((1 - \bar{\alpha}_t) \mI\right)^{-1} \sqrt{\bar{\alpha}_t} \mP\right)^\dagger \sqrt{\bar{\alpha}_t} \mP^\top \left((1 - \bar{\alpha}_t) \mI\right)^{-1} \xbar_t \\
    & = \left(\frac{\bar{\alpha}_t}{1 - \bar{\alpha}_t} \mP^\top \mP\right)^\dagger
    \frac{\sqrt{\bar{\alpha}_t}}{1 - \bar{\alpha}_t} \mP^\top \xbar_t \\
    & = \frac{1 - \bar{\alpha}_t}{\bar{\alpha}_t} \mP^\dagger \left(\mP^\top\right)^\dagger
    \frac{\sqrt{\bar{\alpha}_t}}{1 - \bar{\alpha}_t} \mP^\top \xbar_t \\
    & = \frac{1}{\sqrt{\bar{\alpha}_t}} \mP \mP \mP \xbar_t  = \frac{1}{\sqrt{\bar{\alpha}_t}} \mP \xbar_t.
\end{align*}
The last two equalities hold because ${\mP^\dagger = \mP^\top = \mP = \mP^2}$.
Finally, the right-hand-side in \autoref{eq:weighted-gsure} becomes
\begin{align*}
    & \bbE\left[ \left\lVert \mW \mP^\dagger \mP \left(f(\xbar_t) - \frac{1}{\sqrt{\bar{\alpha}_t}} \mP \xbar_t \right) \right\rVert_2^2 \right]
    + 2 \bbE\left[ \nabla_{\mP^\top \left((1 - \bar{\alpha}_t) \mI\right)^{-1} \xbar_t} \cdot \mW^2 \mP^\dagger \mP f(\xbar_t) \right]
    + c \\
    \stackrel{1}{=}\  & \bbE\left[ \left\lVert \mW \mP \left(f(\xbar_t) - \frac{1}{\sqrt{\bar{\alpha}_t}} \xbar_t \right) \right\rVert_2^2 \right]
    + 2 \bbE\left[ \nabla_{\mP \left((1 - \bar{\alpha}_t) \mI\right)^{-1} \xbar_t} \cdot \mW^2 \mP f(\xbar_t) \right] + c \\
    =\  & \bbE\left[ \left\lVert \mW \mP \left(f(\xbar_t) - \frac{1}{\sqrt{\bar{\alpha}_t}} \xbar_t \right) \right\rVert_2^2 \right]
    + 2 \bbE\left[ \nabla_{\left(1/\left(1 - \bar{\alpha}_t\right)\right) \mP \xbar_t} \cdot \mW^2 \mP f(\xbar_t) \right]
    + c \\
    \stackrel{2}{=}\  & \bbE\left[ \left\lVert \mW \mP \left(f(\xbar_t) - \frac{1}{\sqrt{\bar{\alpha}_t}} \xbar_t \right) \right\rVert_2^2 \right]
    + 2 \bbE\left[ \left(1 - \bar{\alpha}_t\right) \nabla_{\mP \xbar_t} \cdot \mW^2 \mP f(\xbar_t) \right]
    + c \\
    \stackrel{3}{=}\  & \bbE\left[ \left\lVert \mW \mP \left(f(\xbar_t) - \frac{1}{\sqrt{\bar{\alpha}_t}} \xbar_t \right) \right\rVert_2^2 \right]
    + 2 \bbE\left[ \left(1 - \bar{\alpha}_t\right) \nabla_{\xbar_t} \cdot \mP \mW^2 f(\xbar_t) \right]
    + c \\
    \stackrel{4}{=}\  & \bbE\left[ \left\lVert \mW \mP \left(f(\xbar_t) - \frac{1}{\sqrt{\bar{\alpha}_t}} \xbar_t \right) \right\rVert_2^2
    + 2 \left(1 - \bar{\alpha}_t\right) \nabla_{\xbar_t} \cdot \mP \mW^2 f(\xbar_t)
    + c \right],
\end{align*}
which is identical to the expectation term in \autoref{eq:gsure-diff-loss} with $\lambda_t = 1 - \alpha_t$. Justifications:
\begin{enumerate}
    \item ${\mP^\dagger = \mP^\top = \mP = \mP^2}$.
    \item Using the change of variables formula.
    \item Diagonal matrices (such as $\mP$ and $\mW$) commute with one another. Additionally, the divergences w.r.t. $\xbar_t$ and w.r.t. $\mP \xbar_t$ are identical, because $\mP \mW^2 f(\xbar_t)$ equals zero in entries where $\mP$ is zero, and $\mP \xbar_t$ and $\xbar_t$ are identical in entries where $\mP$ is non-zero.
    \item Using the linearity of the expectation operator.
\end{enumerate}
By rewriting both sides of \autoref{eq:weighted-gsure}, we obtain that \autoref{eq:gsure-diff-loss} equals \autoref{eq:projected-diffusion-loss}.
\end{proof}

\section{Detailed Data Descriptions}
\label{sec:data}

\subsection{Dataset Collection}
Here, we detail the collection process for the training and testing data in our experiments.
Note that the data described here is what we consider pristine uncorrupted data.
The corruption process for training GSURE-Diffusion is detailed in \autoref{sec:apdx-deg}.

\paragraph{CelebA.}
In our experiments on human face images,
we use images from the CelebA~\cite{liu2015celeba} dataset.
The original CelebA images were center-cropped to $128 \times 128$ pixels, then resized to $32 \times 32$ pixels, and finally turned into grayscale.
The images were converted to grayscale by averaging all color channels.
Overall, the dataset includes $162770$ training set images, and $19867$ validation set images (which we use for FID~\cite{fid} evaluations).

\paragraph{FastMRI.}
We consider all single-coil knee MRI scans from the fastMRI~\cite{zbontar2019fastmri, knoll2020fastmri} dataset, excluding slices with indices below $10$ or above $40$ as they generally contain less interpretable information.
This yields a training set size of $24853$. For the  validation set we only use the first $1024$ valid slices (which we use for all our post-training experiments).
We treat each slice as a $2$-channel image, separating the complex values into real and imaginary channels.
We center-crop the images to a spatial size of $320\times320$ following~\cite{jalal2021robust}, and normalize the images by $7e-5$ to obtain better neural network performance.
When displaying MR images, we take the absolute value of the complex number in each pixel, and then use min-max normalization to view the resulting values as a grayscale image.
In \autoref{fig:mri_uncertainty}, we jointly normalize standard deviations to ensure a fair visual comparison. We attach a color bar to accurately illustrate the standard deviation intensities.

\subsection{Data Corruptions for GSURE-Diffusion}
\label{sec:apdx-deg}

\paragraph{CelebA.}
In our CelebA~\cite{liu2015celeba} experiments, we consider a degradation operator $\mH$ that randomly drops each $4\times4$-pixel patch with probability $p$.
This operator can be mathematically defined as a diagonal matrix $\mH$ with zeroes in pixels that are dropped and ones in pixels that are kept.
The singular value decomposition (SVD) is trivially and efficiently obtained by
\begin{equation}
    \mH = \mI \mH \mI.
\end{equation}
Note that the singular values in $\mH$ are not ordered.
Since the SVD of $\mH$ has $\mV^\top = \mI$ regardless of the randomness of dropping patches, this family of random operators $\mH$ matches our assumption that all $\mH$ share the same left-singular vectors $\mV^\top$.

Additionally, the projection matrix $\mP = \mH^\dagger \mH$ is simply $\mH$ (as $\mH$ is diagonal with zeroes and ones).
Because each patch is dropped randomly with probability $p$, it follows that $\bbE[\mP] = (1 - p) \mI \succ 0$ is positive definite, matching our assumption. 

Finally, we assume $\mH$ and the additive white Gaussian noise standard deviation $\sigma_0$ to be known for all measurements in the dataset.
For simplicity, we assume a uniform $\sigma_0$ for all measurements.

\paragraph{FastMRI.}
For MRI slices from fastMRI~\cite{zbontar2019fastmri, knoll2020fastmri}, the degradation operator we use is the horizontal frequency subsampling operator used in~\cite{jalal2021robust}.
For an acceleration factor $R$, the degradation operator $\mH$ keeps the central $120/R$ frequencies, and then uniformly samples an additional ${200/R}$ frequencies.
This results in a sampling of $320/R$ frequencies out of the original $320$.
More formally, 
\begin{equation}
\label{eq:mri-svd}
    \mH = \mI \mSigma \mF,
\end{equation}
where $\mF$ is the discrete Fourier transform matrix, and $\mSigma$ is a square diagonal matrix containing ones for frequency indices that are kept by $\mH$, and zeroes elsewhere.
Incidentally, \autoref{eq:mri-svd} is a valid SVD of $\mH$, and can be efficiently simulated using the fast Fourier transform algorithm.

This operator matches our assumptions:
(i) We assume each $\mH$ and the additive white Gaussian noise standard deviation $\sigma_0$ to be known;
(ii) All matrices $\mH$ share the same left-singular vectors defined by $\mF$ (and not depending on the randomness); and
(iii) The central $120/R$ horizontal frequencies are always sampled, and each of the remaining frequencies are equally likely to be sampled, with probability ${200}/({320R - 120})$. Thus $\bbE[\mP]=\bbE[\mSigma^\dagger \mSigma]$ is a diagonal matrix with nonzero diagonal values, making it positive definite.


\section{Implementation Details}
\label{sec:impl}

Our experiments were conducted using DDPM~\cite{ho2020denoising} U-Net architecture with base channel width $128$. All networks were trained using the Adam optimizer, 
dropout with probability $0.1$, EMA with decay factor of $0.9999$.
The diffusion process considered in training for all experiments has $1000$ timesteps, with a linear $\beta$ schedule ranging from $\beta_1 = \sigma_0^2$ ($\sigma_0^2$ is the variance of the AWGN in the data) to $\beta_{1000} = 0.2$.
All experiments were conducted on $8$ NVIDIA A40 GPUs.

In the human faces experiment we ignore the weighting matrix $\mW$ during training because the probability for each pixel to be masked is uniform.
For the knee MRI experiment the weighting matrix $\mW$ was set to $1$ for the centeral lines that were not masked by $\mH$, and $\sqrt{5.8}$ for all other lines matching their inverse square root masking probability (for $R=4$).

All models, including the oracle ones, were trained with the hyperparameters listed in \autoref{tab:hyperparams}.
The ``mean type'' hyperparameter refers to whether the neural network predicts the image $\rvx$ or the added noise $\rveps = \left(\rvx_t - \sqrt{\bar{\alpha}_t} \rvx\right)/\left(\sqrt{1 - {\bar{\alpha}_t}}\right)$.

\begin{table}[ht]
  \caption{Architecture and training hyperparameters for CelebA~\cite{liu2015celeba} and fastMRI~\cite{zbontar2019fastmri, knoll2020fastmri} experiments.}
  \label{tab:hyperparams}
  \begin{center}
  
  \begin{tabular}{l c c}
     &
    \textbf{CelebA} &
    \textbf{FastMRI} \\
    \hline
    \textbf{Iterations} & $180$,$000$ & $31$,$000$ \\
    \textbf{Batch Size} & $128$ & $32$ \\
    \textbf{Learning Rate} & $5e-5$ & $1e-5$ \\
    \textbf{Mean Type} & \texttt{predict\_x} & \texttt{predict\_epsilon} \\
    \textbf{Channel Multipliers} & $[1, 2, 2, 2, 4]$ & $[1, 1, 2, 2, 4, 4]$ \\
    \textbf{Attention Resolutions} & $[16]$ & $[20]$ \\
    $\gamma_t$ & $1$ & $\frac{\bar{\alpha}_t}{1 - \bar{\alpha}_t}$ \\
    $\lambda_t$ & $0.0001$ & $0.0001 \cdot \frac{1 - \bar{\alpha}_t}{\bar{\alpha}_t}$ \\

    
  \end{tabular}
  \end{center}
\end{table}

For the MRI model, we apply an inverse Fourier transform and a Fourier transform to the network's input and output respectively, to utilize the convolutional architecture's advantage on image data (rather than frequencies).
Due to the orthogonality and linearity of the Fourier transform and its inverse, the additive white Gaussian noise remains so, and maintains the same variance.

\Cref{alg:training} shows the GSURE-Diffusion training algorithm used in this paper. We provide our anonymized code and configuration files in the supplementary material.
We intend to publish the code along with our trained model checkpoints upon acceptance.

\begin{algorithm}
\DontPrintSemicolon
\SetKwFor{For}{for}{}{end for}%
 \KwIn{a dataset $\mathcal{D}$ of corrupted images $\mathbf{y}$ with known degradation matrices $\mathbf{H}$ and noise amplitudes $\sigma_0$, learning rate $\eta$, hyperparameters $\lambda_t, \bar{\alpha}_t, T$, and DNN $f_{\theta}^{(t)}(\bar{\mathbf{x}}_t)$ with initial parameters $\theta$.}
 Initialize $\bar{\mathcal{D}} \leftarrow \{\}$ \tcp*{precompute measurements before training}
 \For{$\mathbf{y}, \mathbf{H}, \sigma_0$ in $\mathcal{D}$}{
   $\mathbf{U}, \boldsymbol{\Sigma}, \mathbf{V}^\top \leftarrow \texttt{SVD}(\mathbf{H})$ \\
   $\mathbf{P} \leftarrow \boldsymbol{\Sigma}^\dagger \boldsymbol{\Sigma}$ \\
   $\bar{\mathbf{y}} \leftarrow \boldsymbol{\Sigma}^\dagger \mathbf{U}^\top \mathbf{y}$ \\
   store $\bar{\mathbf{y}}, \boldsymbol{\Sigma}, \mathbf{P}, \sigma_0$ in $\bar{\mathcal{D}}$
 }
 $\mathbf{W} \leftarrow \mathbb{E}_{\mathbf{P} \sim \bar{\mathcal{D}}}\left[ \mathbf{P} \right]$ \\
 \For{$N$ epochs \tcp*[f]{training loop}}{
   \For{$\bar{\mathbf{y}}, \boldsymbol{\Sigma}, \mathbf{P}, \sigma_0$ in $\bar{\mathcal{D}}$}{
     $t \sim \mathcal{U}[1, T], \boldsymbol{\epsilon}_t \sim \mathcal{N}(0, \mathbf{I})$ \\
     $\bar{\mathbf{x}}_t \leftarrow \sqrt{\bar{\alpha}_t} \bar{\mathbf{y}} + \left(\left(1 - \bar{\alpha}_t\right) \mathbf{I} - \bar{\alpha}_t \sigma_0^2 \boldsymbol{\Sigma}^\dagger \boldsymbol{\Sigma}^{\dagger\top}\right)^{\frac{1}{2}} \boldsymbol{\epsilon}_t$ \\
     $\mathcal{L}(\theta) \leftarrow \left\lVert \mathbf{W P} \left( f_{\theta}^{(t)}(\bar{\mathbf{x}}_t) - \bar{\mathbf{y}} \right) \right\rVert_2^2 + 2 \lambda_t \left( \nabla_{\bar{\mathbf{x}}_t} \cdot \mathbf{P} \mathbf{W}^2 f_{\theta}^{(t)}(\bar{\mathbf{x}}_t) \right)$ \\
    $\theta \leftarrow \theta - \eta \nabla_\theta \mathcal{L}(\theta)$
   }
 }
 \KwOut{$\theta$ \tcp*{trained network parameters}}
 \caption{(\textbf{V3qg}) GSURE-Diffusion training algorithm}
\label{alg:training}
\end{algorithm}

\section{Pragmatic Loss Function Considerations}
\label{sec:pragmatic}

\subsection{Divergence Term Estimation}
\label{sec:divergence}
The GSURE-Diffusion training loss in \autoref{eq:gsure-diff-loss} contains a divergence term, which is highly expensive to accurately obtain, in both memory consumption and computation time.
Similar to other SURE-based methods~\cite{metzler2018unsupervised, soltanayev2018training, aggarwal2022ensure}, we use an unbiased Monte Carlo approximation~\cite{ramani2008monte} of the divergence.
Considering the divergence as the trace of the Jacobian matrix $\mJ$ of the term being differentiated ($\mP \mW^2 f_\theta^{(t)}(\xbar_t)$) , Monte Carlo SURE~\cite{ramani2008monte} uses Hutchinson’s trace estimator~\cite{hutchinson1989stochastic}.
We compute the estimate by sampling a random Gaussian vector $\rvv \sim \gN(0, \mI)$, and calculating $\rvv^\top \mJ \rvv$ using automatic differentiation tools.
Notably, this differs from previous methods~\cite{metzler2018unsupervised, soltanayev2018training, aggarwal2022ensure} that used numerical estimates for differentiation, which may suffer from numerical inaccuracies.

\subsection{MSE Term Variance}
\label{sec:mse-variance}

The GSURE-Diffusion loss function in \autoref{eq:gsure-diff-loss} contains the following squared error term
\begin{equation*}
\left\lVert \mW \mP \left(  f_\theta^{(t)}(\xbar_t) - \frac{1}{\sqrt{\bar{\alpha_t}}} \xbar_t \right) \right\rVert_2^2.
\end{equation*}
Because of the possibly strong noise-to-signal ratio present in $\xbar_t$, this term may suffer from high variance, effectively impeding the training process.
To alleviate this, we propose replacing $\frac{1}{\sqrt{\bar{\alpha_t}}} \xbar_t$ with the less noisy $\ybar$, resulting in the loss function 
\begin{equation}
\label{eq:gsure-diff-loss-updated}
    \sum_{t=1}^{T} \gamma_t
    \bbE \left[ \left\lVert \mW \mP \left(  f_\theta^{(t)}(\xbar_t) - \ybar \right) \right\rVert_2^2
    + 2 \lambda_t \left( \nabla_{\xbar_t} \cdot \mP \mW^2 f_\theta^{(t)}(\xbar_t) \right)
    + c \right].
\end{equation}
Note that the difference between the expectations in \autoref{eq:gsure-diff-loss} and \autoref{eq:gsure-diff-loss-updated} is negligible, while \autoref{eq:gsure-diff-loss-updated} has significantly less variance (as $\ybar$ is less noisy than $\frac{1}{\sqrt{\bar{\alpha_t}}} \xbar_t$).
From \autoref{eq:noise-add} we get 
\begin{equation*}
    \frac{1}{\sqrt{\bar{\alpha_t}}} \xbar_t = \ybar + \frac{1}{\sqrt{\bar{\alpha_t}}} \left(\left(1 - \bar{\alpha}_t\right) \mI - \bar{\alpha}_t \sigma_0^2 \mSigma^\dagger \mSigma^{\dagger\top}\right)^{\frac{1}{2}} \rveps_t.
\end{equation*}
We denote $\bar{\rveps} = \frac{1}{\sqrt{\bar{\alpha_t}}} \left(\left(1 - \bar{\alpha}_t\right) \mI - \bar{\alpha}_t \sigma_0^2 \mSigma^\dagger \mSigma^{\dagger\top}\right)^{\frac{1}{2}} \rveps_t$, and show that
\begin{align*}
& \left\lVert \mW \mP \left(  f_\theta^{(t)}(\xbar_t) - \frac{1}{\sqrt{\bar{\alpha_t}}} \xbar_t \right) \right\rVert_2^2 \\
=\  & \left\lVert \mW \mP \left(  f_\theta^{(t)}(\xbar_t) - \ybar - \bar{\rveps}\right) \right\rVert_2^2 \\
=\  & \left\lVert \mW \mP \left(  f_\theta^{(t)}(\xbar_t) - \ybar\right) \right\rVert_2^2
+ \left\lVert \mW \mP \bar{\rveps} \right\rVert_2^2
- 2 \bar{\rveps}^\top \mP \mW \mW \mP \left(  f_\theta^{(t)}(\xbar_t) - \ybar\right) \\
=\  & \left\lVert \mW \mP \left(  f_\theta^{(t)}(\xbar_t) - \ybar\right) \right\rVert_2^2
- 2 \bar{\rveps}^\top \mP \mW^2 \mP  f_\theta^{(t)}(\xbar_t)
+ 2 \bar{\rveps}^\top \mP \mW^2 \mP \ybar
+ \left\lVert \mW \mP \bar{\rveps} \right\rVert_2^2. \\
\end{align*}
The final two terms are constants w.r.t. $\theta$.
Effectively, this means that the difference between the squared error terms in \autoref{eq:gsure-diff-loss} and \autoref{eq:gsure-diff-loss-updated} is $2 \bar{\rveps}^\top \mP \mW^2 \mP  f_\theta^{(t)}(\xbar_t)$.
Under the manifold hypothesis, if $f_\theta^{(t)}(\xbar_t)$ outputs valid images residing on the manifold, and because $\bar{\rveps}$ is a random Gaussian vector, $f_\theta^{(t)}(\xbar_t)$ and $\bar{\rveps}$ are perpendicular.
Therefore, 
the expected difference between the squared error terms, $\bbE\left[ 2 \bar{\rveps}^\top \mP \mW^2 \mP  f_\theta^{(t)}(\xbar_t) \right]$, is zero.
This motivates us to replace \autoref{eq:gsure-diff-loss} with \autoref{eq:gsure-diff-loss-updated}, resulting in significantly lower variance in the loss at little to no cost in terms of bias.

\begin{figure}[t]
    \centering
    \includegraphics[width=0.8\textwidth]{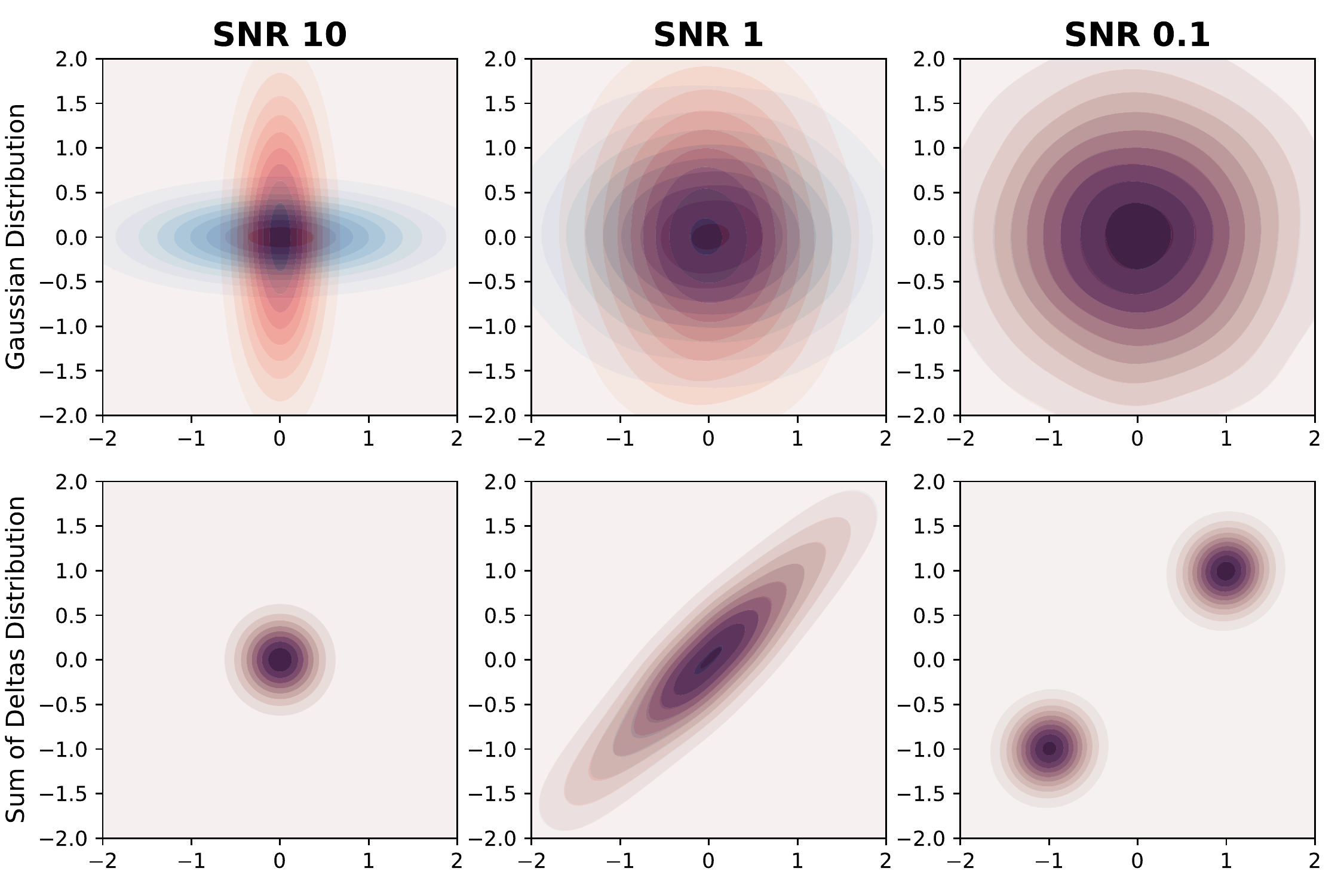} 
    \caption{Analytic example of independence assumptions. The two-dimensional distributions depicted in the graph consist of an isotropic Gaussian distribution (top) and a sum of two delta functions with equal probabilities centered around (1, 1) and (-1, -1) respectively (bottom). The graph illustrates the KDE of the error vectors from the denoiser, with two distinct colors representing two different types of data-point degradations. The blue color represents the denoiser error when the first element is masked, while the red color indicates the denoiser error when the last element is masked. At a low SNR, the KDE of both distributions appears identical when masking is applied to either element. As SNR increases, this independence is weakened in the isotropic Gaussian setting, but remains for the sum of deltas distribution.}
    \label{fig:analytic-example}
\end{figure}

\subsection{Assessment of Generalization Assumption}
\label{sec:generalization}
According to \autoref{thm:ensemble}, the GSURE-Diffusion loss function is equivalent to an MSE objective for samples $\xbar_t \sim \bbE_{\mP}[q(\xbar_t | \xbar, \mP)|\mP]$. During inference, we assume the network is capable of denoising samples from $\xbar_t \sim q^*(\xbar_t | \xbar)$. To asses the generalization capabilities of the network on such samples, we compare the output of a network trained with GSURE-Diffusion and with an oracle network trained on uncorrupted data. In \autoref{fig:generalization} we present the PSNR between the network outputs for denoising samples from $\xbar_t \sim q^*(\xbar_t | \xbar) = q(\xbar_t | \xbar, \mP = \mI)$. The graph shows that for most time-steps, the PSNR is higher than 45dB, leading us to conclude that the model trained with GSURE-Diffusion is an adequate denoiser for samples from $\xbar_t \sim q^*(\xbar_t | \xbar)$, and the assumption is valid.

\subsection{Assessment of Independence Assumption}
\label{sec:independence}
Our proof for \autoref{thm:ensemble} makes use of the assumption made in  ENSURE~\citep{aggarwal2022ensure}, that the network's denoising error is independent of the projection matrix $\mP$. 
To provide some insight into the limitations of this assumption, we begin with a simple example of two-dimensional distributions; one an isotropic multivariate Gaussian and the other a sum of two delta functions, sampled with equal probabilities. From each distribution we sample 10,000 samples $\xbar_t$ using \autoref{eq:noise-add}, where our degradation $\mP$ randomly masks one of the two coordinates with equal probability. We then compute the error $\left(  f_\theta^{(t)}(\xbar_t) - \xbar \right)$ for the two possible projections and examine whether the error is independent of the projection $\mP$. The probabilities of the error vector are depicted in Figure \ref{fig:analytic-example}, using kernel density estimation (KDE). From our analysis, we draw the following observations: In scenarios with high noise levels, the error probabilities are identical for either projection $\mP$, regardless of the distribution. This indicates that much of the information from the original data-point is lost due to the noise, rendering the projection largely ineffectual. For low noise values, the two distributions diverge. The error vectors created from points in the Gaussian distributions are highly correlated with $\mP$, while the ones created from the two deltas remain indistinguishable.  Error vectors generated from points in the Gaussian distribution are highly correlated with $\mP$, whereas those from the two delta functions remain indistinguishable. Hence, we infer that the assumption of independence between the network's denoising error and the projection matrix $\mP$ is dependent upon the underlying dataset and degradation family, necessitating empirical validation before applying our method.

\begin{figure}[t]
    \centering
    \begin{subfigure}[t]{0.49\textwidth}
    \centering\captionsetup{width=0.9\linewidth}
        \includegraphics[width=\textwidth]{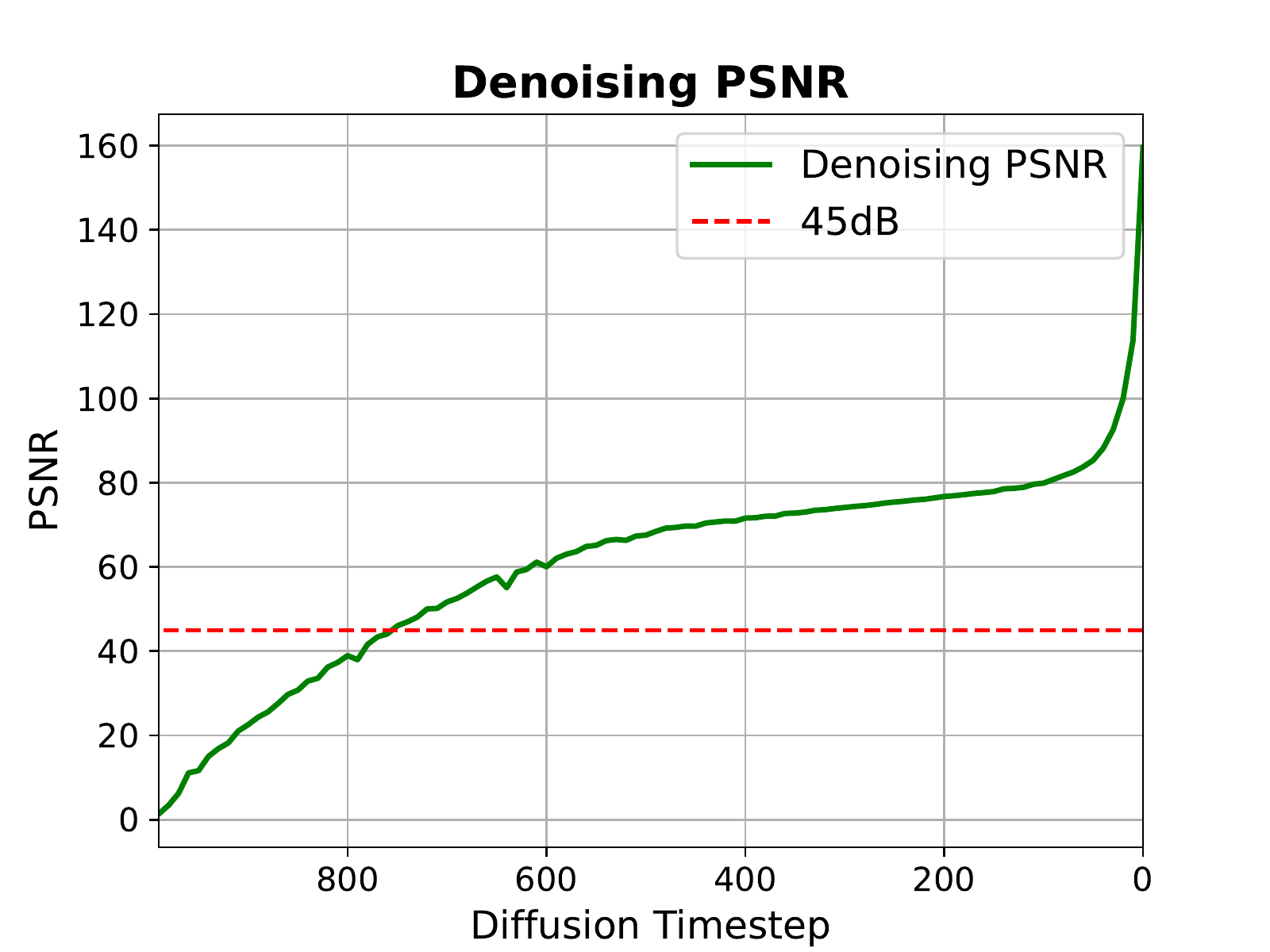}
        \caption{Denoising PSNR between the output of an oracle model and a model trained with GSURE-Diffusion on noisy uncorrupted images.}
        \label{fig:generalization}
    \end{subfigure}
    \begin{subfigure}[t]{0.49\textwidth}
    \centering\captionsetup{width=0.9\linewidth}
        \includegraphics[width=\textwidth]{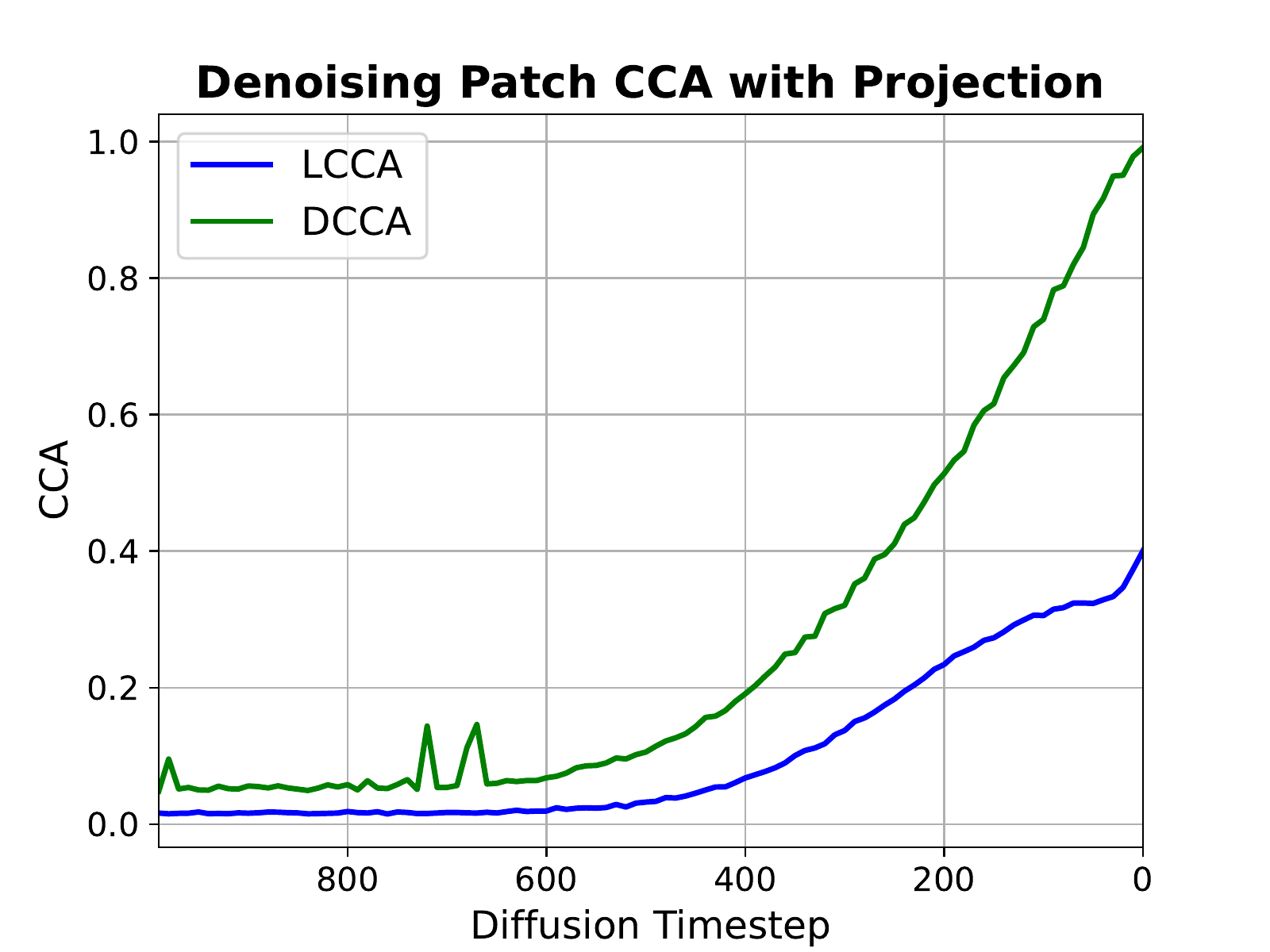}
        \caption{Assessment of correlation between the denoising error and projection, using linear and deep canonical correlation analysis.}
        \label{fig:independence}
    \end{subfigure}
    \caption{Empirical assessment of assumptions.}
\end{figure}

We assess the validity of this assumption for our trained networks using Canonical Correlation Analysis (CCA) measured on corresponding $8\times8$ patches from the denoising error and the projection matrix $\mP$. We perform this comparison using the model trained in \autoref{sec:celeba-exp}, where we can assume locality of the correspondence as the projection matrix $\mP$ is an inpainting binary mask. The model and noise schedule used for the CelebA~\citep{liu2015celeba} assessment fit the model trained with $p=0.2$, $\sigma_0=0.01$.
We present the results in \autoref{fig:independence}, showing both linear and deep CCA (LCCA and DCCA accordingly). As shown in the graphs, the assumption made in ENSURE~\citep{aggarwal2022ensure} nearly holds for our model for timesteps above $500$. The assumption grows less accurate the lower the timestep used in training. This inaccuracy suggests that there may be room for improvement in the assumptions made in ENSURE~\citep{aggarwal2022ensure}, which we believe may be interesting for future work.

\section{Additional Results}

We repeat the human face experiment in \autoref{sec:celeba-exp} with color images from CelebA64~\citep{liu2015celeba}. Both models are trained with the parameters listed in \autoref{sec:impl} for the human face experiment, however the models were trained for a longer 250,000 iterations corresponding to the higher dimension of the data. The finding, shown in \autoref{tab:celeba-color-fid}, from the 10,000 image FID~\citep{fid} with the validation set reflect those in \autoref{sec:celeba-exp}.

\begin{table}[h]
  \caption{FID~\citep{fid} results for diffusion models trained degraded data for color $32 \times 32$-pixel CelebA~\citep{liu2015celeba} images, with different DDIM~\citep{song2020denoising} steps at generation time. Models were trained with (top) or without (bottom) GSURE-Diffusion, on degraded data.}
  \label{tab:celeba-color-fid}
  \begin{center}
  \begin{tabular}{l l c c c c}
    \textbf{Training Scheme} &
    \textbf{Data Degradation} &
    $\mathbf{10}$ \textbf{Steps} &
    $\mathbf{20}$ \textbf{Steps} &
    $\mathbf{50}$ \textbf{Steps} &
    $\mathbf{100}$ \textbf{Steps} \\
    \hline
    Regular
    & No degradation (oracle) & $17.40$ & $11.18$ & $08.32$ & $07.21$ \\
    \hline
    GSURE-Diffusion
    & $p=0.2$, $\sigma_0=0.01$ & $17.39$ & $11.89$ & $09.31$ & $08.97$ \\
  \end{tabular}
  \end{center}
\end{table}


\end{document}